\documentclass[a4paper,USenglish,cleveref,anonymous,autoref, thm-restate, numberwithinsect]{article}

\usepackage{fullpage}
\usepackage{times}

\usepackage{amssymb,amsmath,amsthm}

\usepackage{tikz}
\usetikzlibrary{decorations.pathreplacing}
\usepackage{caption}
\usepackage{subcaption}
\usepackage{float}

\def\firstcircle{(150:1.75cm) circle (2.5cm)}
\def\secondcircle{(30:1.75cm) circle (2.5cm)}
\def\thirdcircle{(270:1.75cm) circle (2.5cm)}
\def\thirdcircleup{(270:1.35cm) circle (2.5cm)}

\newcommand{\C}{{\mathrm C}} % Kolmogorov complexity
\newcommand{\Inf}{{\mathrm I}} % mutual information
\DeclareMathOperator{\lep}{\le^{\lg}}
\DeclareMathOperator{\gep}{\ge^{\lg}}
\DeclareMathOperator\eqp{\mathrel{\stackrel{\mbox{\normalfont\tiny $\lg $}}{=}}}
\DeclareMathOperator\eqdelta{\mathrel{\stackrel{\mbox{\normalfont\tiny $\delta $}}{=}}}

\usepackage{amsmath, amsthm, amssymb, amsfonts, mathtools, comment}
\usepackage{dblfloatfix}
\providecommand{\keywords}[1]{\textbf{\textit{Keywords:}} #1}

\theoremstyle{plain}
\newtheorem{theorem}{Theorem}
\newtheorem{lemma}{Lemma}
\newtheorem{proposition}{Proposition}
\newtheorem{corollary}{Corollary}

\theoremstyle{definition}
\newtheorem{definition}{Definition}

\theoremstyle{remark}
\newtheorem{example}{Example}
\newtheorem{remark}{Remark}

\usepackage{wrapfig}
\usepackage{multicol}

\title{Spectral approach to the communication complexity of multi-party key agreement}

\author{Geoffroy Caillat-Grenier and Andrei Romashchenko}

\begin{document}

\maketitle

\begin{abstract}
We propose a linear algebraic method, rooted in the spectral properties of graphs, that can be used to prove lower bounds in communication complexity.
Our proof technique  effectively marries  spectral bounds with information-theoretic inequalities. 
The key insight is the observation that, in specific settings,  even when data sets $X$ and $Y$ are closely correlated and have high mutual information, 
the owner of $X$ cannot convey a reasonably short message that maintains substantial mutual information with $Y$.
In essence, from the perspective of the owner of $Y$, any sufficiently brief message $m=m(X)$ would appear nearly indistinguishable from a random bit sequence.

We employ this argument in several problems of communication complexity.
Our main result concerns cryptographic protocols.
We establish  a lower bound for communication complexity of multi-party secret key agreement with unconditional, i.e., information-theoretic security. 
Specifically, for  one-round  protocols (simultaneous messages model) of secret key agreement with three participants we obtain an asymptotically  tight lower bound.
This bound implies optimality of the previously known \emph{omniscience} communication protocol
(this result applies to a non-interactive secret key agreement with three parties and  input data sets with an arbitrary symmetric information profile).

We consider communication problems in one-shot scenarios when the parties' inputs are not produced by any i.i.d. sources, and there are no ergodicity assumptions on the input data. 
In this setting, we found it natural to present our results using the framework of Kolmogorov complexity.
\end{abstract}

\keywords{communication complexity,
 Kolmogorov complexity,
 information-theoretic cryptography,
multiparty secret key agreement,
expander mixing lemma,
information inequalities}

\section{Introduction}
\label{sec:intro}

Within computer science, a broad range of communication complexity problems has been studied in recent decades.
In these problems several  (two or more) agents  solve together some task (compute a function,  search an elements in a set, sample a distribution, and so on)
when the input data are distributed among the agents. In different context we may impose different constraints on the class of admissible protocols (protocols can be deterministic or randomized, one-way or interactive, 
with a one shot of  simultaneous messages or with several rounds, etc.). The cost of a communication protocol is the total number of bits that must be exchanged between participants, typically in the worst-case situation.

In this paper we focus on communication problems with three parties (Alice, Bob, and Charlie),  though our techniques can be extended to bigger number of participants.   
We deal with the situation when the input data accessible to Alice, Bob, and Charlie are correlated. In a popular model \emph{number-on-forehead}, the datasets given to Alice, Bob, and Charlie have large intersections, which is a very particular form of correlation between the data.
We study a more general setting (more usual in cryptography and information theory) where the input data sets given to the parties have large mutual information, but it might be impossible to materialize this mutual information as common chunks of bits shared by several parties.

The principal communication problem under consideration is \emph{secret key agreement}: Alice, Bob, and Charlie use the correlation between their input data sets to produce a common secret key. A special feature of this setting is the implicit presence of another participant in the game, Eve (eavesdropper/adversary). 
The eavesdropper can intercept all messages between Alice, Bob, and Charlie, but this should not give Eve any information about the final result of the protocol --- the produced secret key.
A secret key agreement (for two or many participants) is one of the basic primitives in cryptography;  it can serve as a part of more sophisticated protocols (the produced secret key can be used in  a one-time pad encryption  or in 
 more complicated cryptographic schemes).

In practice, the most standard and well known method of secret key agreement is the Diffie-Hellman key exchange \cite{diffie-hellman,merkle} and its generalizations, see \cite{diffie-hellman-generalisations}.
The security of this protocol is based on assumptions of computational complexity.
In particular, the Diffie-Hellman scheme is secure only if the eavesdropper cannot solve efficiently the problem of discrete logarithms. 
Such an assumption looks plausible for most practical applications.
However, theoretical cryptography studies also secret key agreement in information-theoretic settings, where we impose no restrictions on the computational power of the
eavesdropper.  
Besides a natural theoretical interest, such a scheme  can be useful as a building block in more complex protocols.
In particular,  a protocol of information-theoretic secret key agreement  (pretty conventional, involving  communication and computational tools conceivable in the framework of the classical physics)   
is an indispensable  component  of the  protocol of  quantum key distribution (\cite{quantum0,quantum1,quantum2}). %[for arxiv], see also Example~\ref{ex:quantum} below.
\begin{example}\label{ex:quantum}
Let us recall that the standard protocols of quantum key distribution (see, e.g., \cite{quantum0})
can be subdivided into two phases: in the first one,  two parties use a quantum communication channel and  quantum measurements  to produce on both ends a pair of preliminary results that look like two 
strongly correlated but not identical sequences of random bits;  in the second phase, the parties use a classical communication channel  and purely classical computations to perform some sanity check 
and  make sure that the quantum communication was not compromised, and 
then extract a common secret key from the pair of correlated sequences of bits produced in the quantum phase. 
The last part of this scheme is exactly  an  information-theoretic secret key agreement  used in the setting when the two parties are preliminary given  a pair of highly correlated inputs. 
The size of the shared secret key   produced in the classical phase of the protocol depends on the rate of  correlation between the sequences generated by the parties in the quantum phase.
\end{example}
%[for arxiv]Protocols of information-theoretic secret-key agreement  can be useful even in the setting where computational security would suffice.

Besides quantum cryptography,  secret-key agreement based on correlated information appears in various cryptographic  schemes connected with  noisy data 
(biometric information, observations of an inherently noisy communication channel or other physical phenomenon, see  the discussions in \cite{biometric-survey,biometric}),  
in the bounded-storage model 
%\cite{bounded-storage-model-1}
(\cite{bounded-storage-model-2,bounded-storage-model-3}), and so on. 
We refer the reader to the survey  \cite{information-theoretic-security} for a more detailed discussion.  

In  the Diffie-Hellman scheme, the parties may start the protocol from zero, holding  initially  no secret information.  
In contrast, a secret key agreement with  information-theoretic secrecy is impossible 
if the parties  start from scratch. 
To produce a key that is secret in information-theoretic sense, the participants of the protocol  need to be given some input data (inaccessible to the eavesdropper). 
The pieces of input data provided to the parties must be  correlated with each other, and the measure of this correlation determines the optimal size of the common secret key that can be produced.

So far we were very informal and did not specify the mathematical definitions behind the words \emph{secrecy} (of the key) and \emph{correlation} (between parties' inputs).   
Let us describe the settings of information-theoretic secret key agreement more precisely. 
This can be done in different mathematical frameworks. 
  
Historically, information-theoretically secure protocols of secret key agreement were introduced in classical information theory, \cite{ahslwede-csiszar,maurer}.
In this setting, the input data of the parties are produced by correlated random variables. In the settings with two parties it is usually assumed that 
there is a sequence of i.i.d. pairs of random variables with finite range, $(X_i,Y_i)$, $i=1,\ldots,n$, and Alice and Bob receive the values of $(X_1\ldots X_n)$
and $(Y_1\ldots Y_n)$ respectively,
\[
\begin{array}{lcl}
\text{Alice} & \leftarrow & ({ X_1 X_2 \ldots X_n}),\\
\text{Bob}  & \leftarrow  & ({\,Y_1\, Y_2\, \ldots\, Y_n\,}).
\end{array}
\]
Then Alice and Bob run a communication protocol 
and try to produce a common value (secret key) $W$ asymptotically independent of the \emph{transcript} (the transcript consist of  the messages sent by Alice and Bob to each other). 
Ahlswede--Csiszar \cite{ahslwede-csiszar} and Maurer \cite{maurer} found a characterization of the optimal size of $W$ in terms of Shannon's entropy of the input data. 
They showed that the optimal size of the secret key is asymptotically  equal to the mutual information between Alice's and Bob's inputs. A similar characterization of the optimal secret key is known for multi-party protocols, with $k\ge 3$ parties, \cite{ multiparty}.
The problem of secret key agreement and a related problem of \emph{common randomness generation} were extensively studied  in the information theory community  and also  
(in somewhat different settings) in theoretical computer science, see, e.g., \cite{sudan1,sudan2} and the survey \cite{sudan-survey}.

\smallskip

In this paper we follow  the paradigm of  building the foundations of cryptography  in the framework of algorithmic information theory, as suggested in a general form in \cite{antunes}
and more specifically for secret key agreement in  \cite{jacm2019,gr2020}. 
In this approach, the information-theoretic characteristics of the data are  defined not in terms of Shannon's entropy but in terms of Kolmogorov complexity.
In this setting, we can talk about properties of \emph{individual} inputs, keys, transcripts, and not about \emph{probability distributions}.
We assume that the parties (Alice, Bob, Charlie) are given as inputs binary strings $x$, $y$, $z$ respectively,
\[
\begin{array}{lcl}
\text{Alice} & \leftarrow& x,\\
\text{Bob} & \leftarrow&y, \\
\text{Charlie} & \leftarrow&z,
\end{array}
\]
and that the parties know the complexity profile of these strings, i.e., the optimal compression rate of these inputs (precisely or at least approximately, see below).
The secrecy of the produced key means that this key must be incompressible, even conditional on  the public data including the  transcript of the communication protocol.
In other words,  the mutual information (in the sense of Kolmogorov complexity) between the key and the messages sent via the communication channel 
(the transcript) must be negligibly small. 
Practically, this property guarantees  that  the adversary can crack an encryption scheme based on this key  only by the brute-force search,  see the discussion in \cite{gr2020}.
\begin{remark}
The approach based on Kolmogorov complexity seems  more general since we do not need to assume that inputs have any property of stationarity or ergodicity, we do not fix in advance the probability distribution of the pairs of inputs, we do not even assume the existence of such a distribution.  
However, the frameworks of Shannon and Kolmogorov for the definition of secrecy have similar practical interpretations. 
Indeed, a distribution $W$ on $\{0,1\}^n$ has a high entropy, i.e.,  $H(W)\approx n$, if and only if
with a high probability  $W$ returns an $n$-bit string with Kolmogorov complexity  close to $n$  (a random source with a high entropy typically produces incompressible values).
For a more detailed discussion of the connection between Shannon's and Kolmogorov's formalism see~\cite{gruenwald-vitanyi}. 
The  formal statements in Kolmogorov's framework are usually stronger than their homologues in Shannon's framework, 
and theorems from the former theory in most cases formally imply the corresponding results from the latter theory, see \cite{jacm2019}.
\end{remark}

A characterization of the optimal size of the secret key in term  of  Kolmogorov complexity was suggested in \cite{jacm2019}. We begin with the case of two parties, see Theorem~\ref{th:algorithmic} below.
In this theorem, a communication protocol is randomized  (we assume that the parties may use a public source of random bits, which is also accessible to the eavesdropper). 
Let $x$ and $y$ stand for inputs of Alice and Bob, $r$ denote the string of bits  produced by a  public source of randomness (used by the parties and accessible to the eavesdropper), and $t$ denote the transcript of the protocol. 
\begin{theorem}[\cite{jacm2019}]\label{th:algorithmic}
(i) For any numbers $k,\ell\in \mathbb{N}$ and $\epsilon,\delta >0$ there exist a randomized communication protocols $\pi_{k,\ell,\epsilon,\delta}$ such that on every pair of input strings $ (x,y) $ 
(of length at most $n$)
satisfying\footnote{Here the term $ \C(x)$ stands for the plain Kolmogorov complexity  of $x$  (optimal compression of $x$),  
the term $ \C(x \mid y)$ stands for  conditional Kolmogorov complexity of $x$  conditional on $y$ (optimal compression of $x$ given advice $y$), and  the notation 
 $ \C(x) \eqdelta k \text{ and } \C(x\mid y)  \eqdelta  \ell$  means that $|\C(x) - k| \le \delta $ and  $|\C(x\mid y) - \ell | \le \delta $.}  
$
\C(x) \eqdelta k \text{ and } \C(x\mid y)  \eqdelta \ell,
$
Alice and Bob with probability $1-\epsilon$ both obtain a result 
$w = w( x,y,r)$  such that
\begin{equation}\label{eq:rz-weak}
[\text{length of } w  \text{ in bits}]=  \C(x) - \C(x\mid y) - O(\delta )  - o(n) \text{ and }
\C(w \mid \langle t,r\rangle) \ge  |w| - o(n)
\end{equation}
(for $n = |x|+|y|$),
which means that the size of the produced secret key is asymptotically  equal to the mutual information between Alice's and Bob's inputs, and the leakage of information on the key to the eavesdropper (who can access the transcript of the protocol $t$ and public randomness $r$) is negligibly small.

(ii) The size of the key in (i) is pretty much optimal:
 no communication protocol can produce a key $w$ longer than $\C(x) - \C(x\mid y) + O(\delta)  + o(n)$ without loosing the property of secrecy 
 $
 \C(w \mid \langle t,r\rangle) \ge [\text{length of } w  \text{ in bits}]  - o(n)
 $
(the size of a secret key cannot be made asymptotically greater than the mutual information between Alice's and Bob's inputs).
\end{theorem}
\begin{remark}
In Theorem~\ref{th:algorithmic},  the values of $k$ and $\ell$ are embedded in the communication protocol $\pi_{k,\ell,\epsilon,\delta}$.
This means that the parties in some sense ``know'' (at least approximately) the values of $\C(x)$  and $ \C(x\mid y)$.
This is similar to the settings of the classical information theory, where the parties ``know'' the probability distribution on random  inputs and can use a suitable protocol.
The theorem is  nontrivial if the approximation rate $\delta =o(n) $ as $n\to \infty$.
\end{remark}
\begin{remark}
The precision in Eq.~(\ref{eq:rz-weak}) in Theorem~\ref{th:algorithmic} can be made tighter:  there exists  a communication protocol which guarantees
\begin{equation}\label{eq:rz-strong}
[\text{length of } w  \text{ in bits}]=  \C(x) - \C(x\mid y) -  O(\delta) - O(\log n) \text{ and }
\C(w \mid \langle t,r\rangle) \ge  |w| - O(1).
\end{equation}
\end{remark}
\begin{comment}
Theorem~\ref{th:classic} deals with ``typical'' inputs produced by ergodic  sources of correlated data (and even a sequence of i.i.d. random variables), while  in Theorem~\ref{th:algorithmic}
we assume no ergodicity  (this is the so-called ``one-shot'' paradigm). 
Moreover, Theorem~\ref{th:algorithmic}  can be applied in the setting when there is no well defined distribution of probabilities on the input.
Besides, Theorem~\ref{th:classic} claims that Alice and Bob typically can produce a ``good'' value of a secret key for \emph{most} randomly chosen pairs of values of $(\bar X,\bar Y)$,
where the probability is taken over the distribution of random inputs $(\bar X,\bar Y)$. 
In Theorem~\ref{th:algorithmic} the claim is different:  Alice and Bob can produce a ``good'' key
for \emph{all} pairs of inputs $(x,y)$ (with a specified profile of unconditional and conditional Kolmogorov complexity), and 
the probability  is taken over the external source of randomness of the protocol. 
One can show that  the latter condition is formally stronger than the former one, which means that  Claims (i) and (ii) of Theorem~\ref{th:algorithmic} imply, respectively, 
Claims~(i) and (ii) of Theorem~\ref{th:classic},  see the discussion in \cite{jacm2019}.
\end{comment}

%Theorem~\ref{th:classic} and 
Theorem~\ref{th:algorithmic} can be extended to the multi-party setting, where $k>2$ parties are given correlated data and need to agree on common secret key
communicating via a public channel. 
Let us discuss in more detail the version with $k=3$ participants. We assume now that three parties (Alice, Bob, and Charlie) are involved in the protocol.
They are given inputs $x,y,z$ respectively.
We assume that all parties have an access to a common source of  random bits (we denote by $r$ the bits produced by this source)
and exchange messages via a public channel
(we use the conventional definition of a multi-party communication protocol with a public source of random bits, see  \cite{kushilevitz-nisan}).
It is assumed that every message sent by any party reaches every other party (and the eavesdropper). 
In what follows we consider only  triples of inputs $(x,y,z)$ with a ``symmetric'' complexity profile  such that
$\C(x) \approx \C(y) \approx \C(z)$ and $\C(x,y) \approx \C(x,z) \approx \C(y,z)$.
\begin{theorem}[symmetric version of  {\cite[Theorem~5.11]{jacm2019}}]
\label{th:algorithmic-3}
(i) For any profile $(k_1,  k_{2}, k_{3}) \in \mathbb{N}^3$ and $\epsilon,\delta >0$ there exist a randomized communication protocols $\pi_{k_1,k_2,k_3,\epsilon,\delta}$ 
for three parties such that on every triple of binary input strings $ (x,y,z) $   (of length at most $n$)
satisfying
\begin{equation}\label{eq:sym-profile}
\begin{array}{l}
\C(x)  \eqdelta
\C(y)  \eqdelta  
\C(z) \eqdelta k_1,\
\C(x,y)  \eqdelta  
\C(x,z)  \eqdelta 
\C(y,z)  \eqdelta k_2,\ 
\C(x,y,z) \eqdelta  k_{3}
\end{array}
\end{equation}
 Alice, Bob, and Charlie can agree with probability $1-\epsilon$ on a key
$w = w( x,y,z,r)$ such that
\begin{equation}\label{eq:rz-triple}
[\text{length of } w  \text{ in bits}] = 
\begin{array}{l}
\frac{I(x:y\mid z) +  I(x:z\mid y) + I(y:z\mid z)}2 + I(x:y:z)
\end{array}
 - O(\delta ) -  o(n)
\end{equation}
 (for $n=|x|+|y|+|z|$) and 
\begin{equation}\label{eq:secrecy}
 \C(w \mid \langle t,r\rangle) \ge  |w| - o(n).
\end{equation}

(ii) The size of the key in (i) is  asymptotically optimal, i.e., no communication protocol can give a key $w$
asymptotically longer than 
\begin{equation}\label{eq:key-size}
\frac12\left(I(x:y\mid z) +  I(x:z\mid y) + I(y:z\mid z)  \right) + I(x:y:z) +  O(\delta) + o(n)
\end{equation}
%\eqref{eq:rz-triple} 
without loosing the property of secrecy  \eqref{eq:secrecy}.
\end{theorem}
\begin{remark}
The general version of \cite[Theorem~5.11]{jacm2019} applies to a triple of inputs with arbitrary (possibly non-symmetric) complexity profile.
In the general case, the characterization of the optimal size of the secret key is more involved then \eqref{eq:rz-triple}
and involves piece-wise linear expression involving the terms of the mutual information for $x$, $y$, and $z$, see  \cite{jacm2019}.
We discuss only  symmetric complexity profiles in order to avoid cumbersome formulas and focus on the most essential combinatorial ideas behind the proofs.
\end{remark}

The known proofs of the positive parts of Theorem~\ref{th:algorithmic} and Theorem~\ref{th:algorithmic-3} (the existence of protocols) are quite explicit and constructive: we know specific communication protocols that 
allow to produce a secret key of the optimal size. 
More specifically,  the proofs suggested in \cite{jacm2019} provide a protocol for  Theorem~\ref{th:algorithmic}(i) with communication complexity 
\begin{equation}\label{eq-cc-2}
\min \left\{  \C(x \mid y), \C(y \mid x)\right\} +  O(\delta )  + O(\log n)
\end{equation}
and  a protocol\footnote{%
The scheme proposed in \cite{jacm2019} is the so called \emph{omniscience} protocol.
In this protocol, all parties send simultaneously  their messages (random hash-values of the inputs) so that each of them learns completely the entire triple of inputs $(x,y,z)$
(this explains the term \emph{omniscience}).
The total length of the sent messages is less than $\C(x,y,z)$, so an eavesdropper can learn only a partial information on the inputs. 
The gap between the total complexity of $\C(x,y,z)$ and the divulged information is used to produce a secret key.%  
}
 for  Theorem~\ref{th:algorithmic-3}(i) with communication complexity
\begin{equation}\label{eq-cc-3}
\C(x,y,z) - \frac12\big(I(x:y\mid z) +  I(x:z\mid y) + I(y:z\mid z)  \big)  - I(x:y:z) +  O(\delta )  + O(\log n).
\end{equation}
The communication complexity \eqref{eq-cc-2}  from Theorem~\ref{th:algorithmic}(i) is known 
to by asymptotically optimal, see \cite{gr2020}.
In this paper we study the communication complexity of the problem from Theorem~\ref{th:algorithmic-3}.
In fact, \eqref{eq-cc-3} is \emph{not} optimal for general communication protocols; however, we show that this communication complexity is asymptotically  optimal
in the class of protocols with \emph{simultaneous messages}, i.e., in the model where Alice, Bob, and Charlie send their messages in parallel,  receive the messages sent by their vis-a-vis,
and compute the result (secret key) without any further interaction. 

\begin{theorem}[main result]\label{th:main}
In the setting of Theorem~\ref{th:algorithmic-3}, communication complexity of a protocol with simultaneous messages
(the total number of bits sent by Alice, Bob, and Charlie) for triples of inputs $(x,y,z)$ with a symmetric  complexity profile~\eqref{eq:sym-profile})
cannot be smaller than
\begin{equation}\label{eq-cc-3-lower-bound}
\C(x,y,z) - \frac12\big(I(x:y\mid z) +  I(x:z\mid y) + I(y:z\mid z)  \big)  - I(x:y:z) -  O(\delta)  - O(\log n).
\end{equation}
\end{theorem}
Communication complexity  \eqref{eq-cc-3-lower-bound} is not optimal for general (multi-round) communication protocols of secret key agreement,  see Proposition~\ref{thm:5}.

\smallskip

The proof of our main result combines information-theoretic techniques and spectral bounds for graphs (the expander mixing lemma). 
Spectral bounds \emph{per se} are not new in communication complexity (see, e.g., the usage of Lindsey's lemma in \cite{lindsey-lemma}).
Information-theoretic methods are also  pretty common in this area. 
But the combination of these two techniques seems to be less standard.
The key step  of the proof is the observation  that in some setting, when parties hold correlated data sets, for each of them it is hard to send a message that has non-negligible mutual information with the partners' data. 
In other words,  a ``too short'' message sent by Alice would have zero mutual information with the data $(y,z)$ given to Bob and Charlie. 
For  secret key agreement protocols, this observation implies that the messages of every party inevitably have to be quite long.
A similar argument can be used in problems that are not connected with cryptography, see Theorem~\ref{th1}.

\smallskip

The rest of the paper is organized as follows. In Section~\ref{sec:preliminaries} we recall several standard definitions and introduce the notation.
In Section~\ref{sec:informal-proof} we explain informally the scheme of our argument. 
In Section~\ref{sec:1} we prove the main technical tool of this paper, Theorem~\ref{thm:1} (which claims that in some setting, it is hard to send a message that has non-negligible mutual information with the partners' data).
In Section~\ref{sec:parallel-messages} we illustrate the application of our technique with a simple example that is not related to cryptography.
In Section~\ref{sec:crypto-special-case} we prove Theorem~\ref{th:main}  for a restricted (``the most important'') class of complexity profiles; this is the main technical contribution of the paper.
In Section~\ref{sec:crypto-general-case} we extend this result  and prove Theorem~\ref{th:main} for all (symmetric) complexity profiles.
We conclude  with a discussion of limitations of our technique and open problems.
Several technical lemmas are deferred to Appendix.

\section{Preliminaries and Notation}\label{sec:preliminaries}

\subsection{General notation.}
For a binary string $x$ we denote its length $|x|$. 
For a finite set $S$ we denote its cardinality~$\# S$.

In what follows we manipulate with equalities and inequalities for Kolmogorov complexity. 
Since  many of them hold up to a logarithmic term, we use the notation
$A\eqp B$,  $A\lep B$,  and $A\gep B$
for 
$
|A-B| = O(\log n),\ A\le B + O(\log n), \text{ and } B\le A + O(\log n)
$
respectively, where $n$ is clear from the context ($n$ is usually the length of the strings involved in the inequality).

$\mathbb{F}_q$ denotes the field of $q$ elements (usually $q= 2^n$). 
A $k$-dimensional vector over $\mathbb{F}_q$ is a $k$-tuple  $(x_1,\ldots, x_k) \in \mathbb{F}_q^k$. 
We say that two vectors $(x_1,\ldots, x_k)$ and $(y_1,\ldots, y_k)$ in $ \mathbb{F}_q^k$ are orthogonal to each other if
$
x_1y_1 +\ldots + x_k y_k = 0
$
(the addition and multiplication are computed in the field $ \mathbb{F}_q$).
A vector is called self-orthogonal if it is orthogonal to itself. In a $k$-dimensional space over the field of characteristic $2$ there  are $2^{k-1}$ self-orthogonal vectors $(x_1,\ldots, x_k)$ and they form a linear subspace of co-dimension $1$ (a vector is self-orthogonal iff $x_1+\ldots+x_k = 0$).
A \emph{direction} in $\mathbb{F}_q^k$  is an equivalence class of non-zero vectors over $\mathbb{F}_q$ that are proportional to each other (a direction can be understood as a point in the projective space of dimension $k-1$).
%we specify a direction as $(x_1:\ldots: x_k) $, where $(x_1,\ldots, x_k) $ is one of the vectors in this equivalence class. 
Two directions are orthogonal to each other  if every vector in the first one is orthogonal to every vector in the second one.

$\C(x)$ stands for Kolmogorov complexity of $x$ (the length of the shortest program\footnote{In an optimal programming language, see Appendix for more detail.} producing~$x$) and $\C(x \mid y)$ (the length of the shortest program producing $x$ given input $y$) stands for Kolmogorov complexity of $x$ given $y$. 
Respectively, $\Inf(x:y)$ and $\Inf(x:y\mid z)$ denote the mutual information between $x$ and $y$ and the conditional information between $x$ and $y$ given $z$. We use the notation $\Inf(x:y:z) := \Inf(x:y) - \Inf(x:y\mid z)$.
For a tuple of strings $(x_1,\ldots, x_n)$ its \emph{complexity profile} is the vector consisting of the complexity values $\C(x_{i_1},\ldots, x_{i_s})$ (for all $2^n-1$ sub-tuples $1\le i_1<\ldots <i_s\le n$). 

Kolmogorov complexity can be relativized: $\C^{\cal O}(x)$ and $\C^{\cal O}(x\mid y)$ stand for Kolmogorov complexity of $x$ (conditional on $y$) assuming that the universal decompressor can access  oracle $\cal O$. If the oracle is a finite string $s$, then $\C^{\cal O}(x) = \C(x\mid s) + O(1)$. 

For more detail on the basic facts about Kolmogorov complexity see Appendix. A comprehensive  introduction in the theory of Kolmogorov complexity can be found in  \cite{li-vitanyi} and \cite{shen-vereshchagin}.

\subsection{Communication complexity.}
We use the conventional notion of a communication protocol for two or three parties, see for detailed definitions \cite{kushilevitz-nisan}.
We discuss \emph{deterministic protocols} and \emph{randomized protocols with a public source of random bits} (see Appendix for more detail).

In general, a communication protocol may consist of several rounds, when each next message of every party depends on the previously sent messages. 
In the \emph{simultaneously messages} model there is no interaction: all parties  send in parallel their messages  that depend only on their own input data (and the random bits), and then compute the final result.

We usually denote the inputs of Alice, Bob, and Charlie as $x$, $y$, and $z$ respectively (\emph{number-in-hand} model). 
A deterministic communication protocol for inputs  $x,y,z\in\{0,1\}^n$ returns a result $w=w(x,y,z)$. 
In a randomized protocol the result depends also on the public source of random bits $r$, and $w=w(x,y,z,r)$. 
The sequence of messages sent by the parties to each other while following the steps of the protocol is called a \emph{transcript} $t=t(x,y,z)$ of the communication
($t=t(x,y,z,r)$ for randomized protocols).   
Communication complexity of a protocol is the maximal length of its transcript (measured in bits), i.e., $\max\limits_{x,y,z,r} |t(x,y,z,r)|$. 
%In this paper, all theorems are formulated for randomized communication protocols.

A communication protocol \emph{computing a function} $F(x,y,z)$ returns a correct result if $w(x,y,z,r) = F(x,y,z)$.
For a \emph{secret key agreement} protocol, the definition of a \emph{correct result} $w$ is subtler:  we need that  (i)~$w$ is of the required size and (ii)~it is almost incompressible even given the transcript of the communication $t$ and the public random bits $r$. 
For a more detailed discussion of this setting we refer the reader to \cite{jacm2019}.

We will assume that the communication protocol has a ``uniform'' description.
More technically,  we assume that for $n$-bit inputs (the full description of such a protocol) has an efficient description of size $O(\log n)$. 
For such a protocol we do not loose much security even if the description of the protocol is available to the eavesdropper.
Thus,  we cannot ``cheat''  by embedding in the structure of the protocol any secret information hidden from the adversary.

If the length of inputs is equal to $n$, we may assume w.l.o.g. that the  used string of public random bits $r$ is of length $O(n)$.
(Using longer sources of random bits may only slightly affect the probability of an erroneous result. 
For protocols computing a function, $O(\log n)$ bits is enough due to the Newman's theorem, \cite{newman}; 
in secret key agreement protocols we may need $O(n)$ random bits, see \cite{gr2020}). 
This is crucial for our setting: we may assume that the terms $O(\log \C(r))$  involved in inequalities for Kolmogorov complexity match in order of magnitude the terms $O(\log (\C(x) +\C(y)+\C(z)))$, where $x,y,z$ are the input data.

\subsection{Reminder of the spectral graph  technique.}
Let $G = (L\cup R, E)$ be a bi-regular bipartite  graph where each vertex in $L$ has degree $D_L$, each vertex in $R$ has degree $D_R$, 
and each edge $e\in E$ connects a vertex from $L$ with a vertex from $R$ (observe that $\# E = \#L \cdot D_L = \# R\cdot D_R$). 
The adjacency matrix of such a graph is a  zero-one matrix 
$
M = \left(
\begin{array}{cc}
0 & A\\
A^\top &0
\end{array}
\right)
$
where $A$ is a matrix of dimension $(\# L) \times (\# R)$ 
($A_{xy} =1$ if and only if there is an edge between the $x$-th vertex in $L$ and the $y$-th vertex in $R$). 
Let
$
\lambda_1\ge\lambda_2\ge\ldots\ge\lambda_{N}
$
be the eigenvalues of $M$, where $N=\#L + \#R$ is the total number of vertices. 
Since $M$ is symmetric, all $\lambda_i$ are real numbers. 
It is well known that for a bipartite graph
the spectrum is symmetric, i.e., $\lambda_i= - \lambda_{N-i+1}$ for each $i$, and $\lambda_1=-\lambda_{N} =\sqrt{D_L D_R}$ (see, e.g., \cite{mixing-lemma}). 
The graphs with a large  \emph{spectral gap} (the gap between the first and the second eigenvalues) have the property of good \emph{mixing}, see \cite{graph-spectrum}.
\begin{lemma}[Expander Mixing Lemma for bipartite graphs, see  \cite{mixing-lemma}]  \label{mixing-lemma}
	Let $G=(L\cup R,E)$ be a regular bipartite graph where each vertex in $L$ has degree $D_L$ and  each vertex in $R$ has degree $D_R$. 
	Then for each $A\subseteq L$ and $B\subseteq R$ we have
	$
	\left|E(A,B)-\frac{D_L\cdot \# A \cdot \#B }{\# R} \right|\leq \lambda_2 \sqrt{\# A\cdot \# B},
	$
	where $\lambda_2$ is the second largest 
	eigenvalue of the adjacency matrix of $G$ and $E(A,B)$ is the number of edges between $A$ and $B$.
\end{lemma}
We apply  Lemma~\ref{mixing-lemma}  for the case $\frac{D_L\cdot \# A \cdot \#B }{\# R} \ge  \lambda_2 \sqrt{\# A\cdot \# B}$, as shown in the corollary below.
\begin{corollary}\label{mixing-lemma-large-sets}
Let $G=(L\cup R,E)$ be a graph from  Lemma~\ref{mixing-lemma} with the second eigenvalue $\lambda_2$.  
Then for $A\subseteq L$ and $B\subseteq R$ such that $\# A \cdot \# B \ge \left(\frac{\lambda_2 \#R}{D_L}\right)^2$ we have 
	\begin{equation}
		\label{eq:mixing-easy}
		E(A,B)  = O\left(  \frac{D_L\cdot \# A \cdot \# B}{\#R }  \right).
	\end{equation}
\end{corollary}
To apply the expander mixing lemma, we need a graph with a large spectral gap. 
In particular, we use the following well know lemma.
\begin{lemma}[see \cite{hoffman}]
\label{l:projective-plane}
Let $G = (L\cup R, E)$ be a graph where $L$ consists of all lines in a projective plane over a finite field $\mathbb{F}_q$,
$R$ consists of all points in this plane, and $E$ consists of all incident pairs $(\text{line}, \text{point})$. 
In this graph $\lambda_1 = \Theta(q)$ and $\lambda_2 = O(\sqrt{q})$.
\end{lemma}
In the next section we  discuss spectral properties of graphs with a more complicated structure.

\section{Main technical tools and the scheme of the proof}
\label{sec:informal-proof}

In this section we sketch the proof of our main result (Theorem~\ref{th:main}). 
In this sketch we ignore technical  difficulties that can be resolved with  standard techniques or \emph{ad hoc} tricks, and focus on the main ideas used in the proof.

\smallskip

\noindent
\textbf{3.1. Setting the parameters.} Let us assume that $\delta  = O(\log n)$, 
i.e., all parties of the protocol ``know''  the complexity profile of the triple of inputs $(x,y,z)$ up to an additive logarithmic term\footnote{%
A logarithmic error term is, in some sense, the finest meaningful precision for Kolmogorov complexity. 
All our arguments can be repeated \emph{mutatis mutandis} for any coarser precision $\delta$ such that $\log n \ll \delta(n) \ll  n$.}.
This assumption does not affects significantly the argument, but it helps to avoid minor technical details and makes the explanation more transparent.
To simplify the notation, in this section we discuss only triples of inputs with the profile
\begin{equation}\label{eq:main-profile}
\begin{array}{l}
\C (x) \eqp \C(y) \eqp \C(z) \eqp k n,\
\C(x,y) \eqp \C(x,z) \eqp \C(y,z) \eqp (2k-1) n,\\ 
\C(x,y,z) \eqp (3k-3)n
\end{array}
\end{equation}
which is equivalent to  
\[
\begin{array}{l}
\C(x\mid y,z) \eqp \C(y\mid x,z) \eqp \C(z\mid x,y) \eqp (k-2)n,  \\
I(x:y\mid z)  \eqp I(x:z\mid y)  \eqp I(y:z\mid x) \eqp n,\ 
I(x:y:z) \eqp 0
\end{array}
\]
see Fig.~\ref{fig:profiles}~(b).
			\begin{figure}[h]
    			 \begin{subfigure}[b]{0.5\textwidth}
			 	\begin{flushleft}
				\begin{tikzpicture}[scale=0.49]
				  \draw \firstcircle node[ left] {\small $n$};
				  \draw \secondcircle node [right] {\small $n$};
				  \node at (95:0.90)   {\small $n$};
				  \node at (150:4.75) {\Large $x$};
				  \node at (30:4.75) {\Large $y$};
				\end{tikzpicture}
				\vspace{6em}
				\caption{Complexity profile for Examples~\ref{example:1} and~\ref{example:line-point}: 
				\newline $\C(x\mid y ) \eqp n$, $\C(y\mid x) \eqp n$, $\Inf(x:y) \eqp n$.\newline}
				\end{flushleft}
			\end{subfigure}
			\hspace{2em}
    			 \begin{subfigure}[b]{0.49\textwidth}
			 	\begin{flushleft}
				\begin{tikzpicture}[scale=0.6]
				  \draw \firstcircle node[above left] {\small $(k-2)n$};
				  \draw \secondcircle node [above right] {\small $(k-2)n$};
				  \draw \thirdcircle node [below] {\small $(k-2)n$};
				  \node at (95:0.05)   {\small $0$};
				  \node at (90:1.55) {\small $ n$};
				  \node at (210:1.75) {\small $n$};
				  \node at (330:1.75) {\small $n$};
				  \node at (150:4.75) {\Large $x$};
				  \node at (30:4.75) {\Large $y$};
				  \node at (270:4.75) {\Large $z$};
				\end{tikzpicture}
				\caption{Complexity profile for Proposition~\ref{p:hypergraph}:
				 \newline $\C(x\mid y,z)\eqp (k-2)n$, $\Inf(x:y)\eqp n$, \newline $\Inf(x:yz)\eqp2n$,  $\Inf(x:y:z) \eqp0$.}
				\end{flushleft}
			\end{subfigure}
			\caption{Diagrams  with complexity profiles for Examples~\ref{example:1}-\ref{example:line-point} and Proposition~\ref{p:hypergraph}.}
				\label{fig:profiles}
			\end{figure}
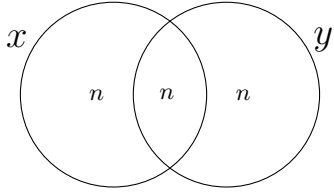
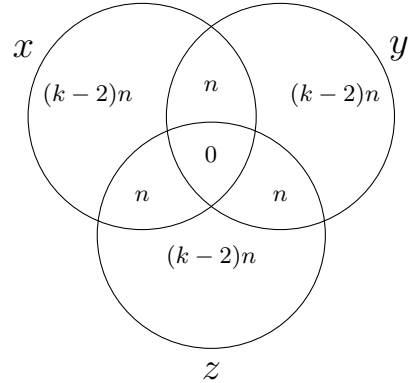
In this setting, Theorem~\ref{th:algorithmic-3} gives the optimal size of a secret key 
\begin{equation} \label{eq:key-size-specific} 
\frac12\big( \Inf(x:y\mid z)  + I(x:z\mid y)  +  I(y:z\mid x)  \big) + \Inf(x:y:z) \eqp1.5n.
\end{equation}
Our aim is to bound communication complexity for inputs with this complexity profile:
\begin{theorem}[special case of Theorem~\ref{th:main}]\label{th:main-special-case}
In the setting of Theorem~\ref{th:algorithmic-3}, communication complexity of a protocol with simultaneous messages
(the total number of bits sent by Alice, Bob, and Charlie) 
for some triples of inputs $(x,y,z)$ with complexity profile~\eqref{eq:main-profile} cannot be smaller than $(3k-4.5)n$, which matches Eq.~\eqref{eq-cc-3-lower-bound}. % for the case~\eqref{eq:main-profile}.
\end{theorem}

\smallskip

\noindent
\textbf{3.2. Preliminary consideration: the need for hard inputs.}
The optimal size of the secret key in Theorem~\ref{th:algorithmic} and Theorem~\ref{th:algorithmic-3} depends only on the complexity profile of $(x,y,z)$ and not on the combinatorial structure of the input. 
The situation with communication complexity (the number of bits sent by the parties)  is different: 
it  may vary significantly for different tuples of inputs with the same complexity profile. 
When we talk about the communication complexity of a protocol, we mean the worst-case complexity, i.e., the maximal number of sent bits among all admissible inputs. 
To prove a lower bound for the worst-case communication complexity, we need to provide a  triple  of inputs  for which the parties have to send long messages. 
We provide a class of inputs that are guaranteed to be ``hard''  (for all valid protocol, for most triples of inputs from this class, communication complexity is high).

\smallskip

\noindent
\textbf{3.3. First step of the argument: conditional on Charlie's message, the mutual information between Alice's and Bob's inputs must increase.} 
We begin with an observation that might  seem to have nothing to do with communication complexity. 
We recall the lower bound for the size of the secret key (that applies to protocols with any communication complexity).
In \cite{jacm2019} (see Theorem~\ref{th:algorithmic}(ii)) it is shown that two parties, Alice and Bob, can agree on secret key of complexity $k$
  \emph{only if} the mutual information between Alice's input $x$ and Bob's input $y$ is greater than $k$. 
 The proof of this statement can be  easily adapted to the following slightly more general  setting:  
\begin{lemma}\label{lemma:1}
Assume that there is a publicly available information $s$ (accessible to Alice, Bob, and the eavesdropper), and besides this information 
Alice is given a private input $x$ and Bob is given a private input  $y$. 
Then, by communication via a public channel accessible to the eavesdropper, Alice and Bob cannot agree on a secret key of complexity  greater than $I(x:y \mid s)$. 
\end{lemma}
We apply this proposition to a protocol with three parties. 
Let $t_C$ denote the concatenation of the messages sent by Charlie. 
This is a piece of publicly available information (accessible to Alice, Bob, and the eavesdropper). 
Due to Lemma~\ref{lemma:1}, Alice and Bob cannot agree on a secret key with Kolmogorov  complexity  greater than   $I(x:y\mid t_C)$ (at this point we ignore  whether Charlie can learn the same key or not). 
Hence,  in the settings  \eqref{eq:main-profile},  a secret key of size \eqref{eq:key-size-specific} can be produced only if 
$
 I(x:y\mid t_C) \gep 1.5n.
$
Observe that in the setting \eqref{eq:main-profile} the mutual information between $x$ and $y$ is equal to $n$. 
This means that the mutual information between Alice's and Bob's inputs \emph{conditional on Charlie's message}, i.e., $ I(x:y\mid t_C)$,  is  bigger than 
the unconditional mutual information between Alice's and Bob's inputs, i.e., $I(x:y)$.  
A pretty standard information-theoretic argument implies that the gap between $I(x:y)$ and $ I(x:y\mid t_C)$ is not greater than  the mutual information between $\langle x,y\rangle$ and $t_C$,  and we conclude that
$
 I(x,y:t_C)  \gep n/2.
 $
In other words, Charlie must send a message  $t_C$ that has $ \ge n/2$ bits of mutual information with the pair of inputs of Alice and Bob.
A similar argument implies that Alice  must send a message $t_A$ such that $ I(y,z:t_A)  \gep n/2$ 
and Bob  must send a message $t_B$ such that $ I(x,z:t_B)  \gep n/2$. 

This part of the argument is based on Lemma~\ref{lemma:1}, which re-employs an argument from \cite{jacm2019} in a pretty direct way. 
So at this stage  we need no  substantially new ideas.

\medskip

\noindent
\textbf{3.4. Second step of the argument: it may be difficult for Alice to send a message increasing the mutual information between Bob's and Charlie's inputs.} 
We have shown above that in the setting \eqref{eq:main-profile} Alice, Bob, and Charlie can agree on a secret  key of optimal size only if each of them sends a messages
that contains $\gep n/2$ bits of mutual information with the inputs of two other parties

We are going to show that this may require sending  \emph{very long} messages (much longer than  $ n/2$ bits). 
This part of the argument is the main technical contribution of our paper. 
To explain this idea, we make a digression and discuss a similar problem in simpler settings.

\smallskip
\noindent
\emph{Digression: how to say something that the interlocutor already knows.}
Let us consider randomized communication protocols with two participants playing  non-symmetric roles. 
We call the participants Speaker and Listener
and assume that  Speaker holds an input string $a$ and Listener holds another input string $b$.
This is a one-way protocol: Speaker  sends a message to Listener in one round, without any feedback.
The aim of Speaker  is to send to Listener a message that is \emph{not completely unpredictable} from the point of view of Listener. 
More precisely, Speaker's message must have positive (and non-negligible) mutual information with Listener's input~$b$. 
We start with a simple example when the task of Speaker is trivial.
\begin{example}\label{example:1}
Let Speaker is given a string $a=uv$ and Listener is given a string $b=uw$, where $u$, $v$, and $w$ are independent incompressible strings of length $n$, i.e., 
$
\C(uvw) \eqp \C(u) + \C(v) + \C(w) \eqp 3n.
$
Observe that 
\begin{equation}\label{eq:profile-line-point}
\C(a)\eqp 2n,\  \C(b)\eqp 2n,\  \Inf(a:b)\eqp n
\end{equation}
(see the diagram in Fig.~\ref{fig:profiles}~(a)).
In this setting, if Speaker wants to  communicate a message of length $n$ with a \emph{high} mutual information with Listener's $y$, she may send a part of $u$,  which is know to both participants of the protocol.
On the other hand, if Speaker wants to  communicate a  message  with a \emph{low} mutual information with Listener's $b$, this is also possible:
Speaker may send a part of $v$,  which is know to Speaker but not to the Listener.
\end{example}
Let us proceed with a  less trivial example. 
\begin{example}\label{example:line-point} 
Now we consider a pair $(a,b)$ with the same complexity profile as in Example~\ref{example:1} but with a different combinatorial structure. 
Let $a$ be a line in the projective plane over the finite field $\mathbb{F}_{2^n}$ and $b$ be a point in the same projective plane incident to $a$, and the pair $(a,b)$ have the maximal possible complexity
(among all incident pairs $(\text{line}, \text{point})$ in the plane).
For these $a$ and $b$ we have the same complexity profile \eqref{eq:profile-line-point}. 
Indeed, we need two elements of the field ($2n$ bits of information) to specify a line or a point, 
but we need only one element of the field ($n$ bits of information) to specify a point when  a line is known. 
However, the combinatorial properties of this pair are very different from the properties of the pair in Example~\ref{example:1}. 

If Speaker is given $a$ and Listener is given $b$ as above, then Speaker cannot  send a \emph{reasonably short} message having non-negligible mutual information with Listener's input $b$.
In fact, if Speaker  wants to send to Listener a message $m = m(a)$ having $\delta$ bits of mutual information with $b$, then the size of $m$ must be at least $n+\delta$.
In particular, if the message $m$ is shorter than $n$, then it cannot contain any information on $b$.
We prove this statement in Section~\ref{sec:1}. 
\end{example}

Example~\ref{example:line-point} is an instance of a much more general phenomenon.
Let us have a bipartite graph $G=(V_L,V_R,E)$, where the set of vertices is $V_L\cup V_R$ and the set of edges is $E\subset V_L\times V_R$. 
We assume that the graph is bi-regular, i.e., all vertices in $V_L$ have the same degree $D_L$ and all vertices in $V_R$ have the same  degree $D_R$ (we always assume that $D_L \ge D_R$). 
We say that $G$ is a \emph{spectral expander}\footnote{We use the term \emph{expander} without assuming that the degree of a graph is constant.} if  the second eigenvalue of its adjacency matrix $\lambda_2 = O(\sqrt{D_L})$. 
Let $(x,y)\in E$ be a ``typical'' edge of this graph (in the sense that its Kolmogorov complexity is close to the maximum possible value), and
let  $x$ and $y$ be the inputs given to Alice and Bob respectively.
Then we have a property  similar to Example~\ref{example:1}:  if Alice wants to send a message having $\delta$ bits of mutual information with Bob's data $y$, she must send a message of size at least $\log D_R + \delta$.
We prove this fact  using the Expander Mixing Lemma. 
(Example~\ref{example:line-point} corresponds to the graph $G=(V_L,V_R,E)$ where $V_L$ consists of all lines in the plane, $V_R$ consists of all points in the plane, and $E$ is the set of all pairs of incident lines and points;
it is known that this graph is a spectral expander.) [End of \emph{Digression}.]

\smallskip

Now we  generalize the observations from  the \emph{Digression} above and explain the main idea of the proof of Theorem~\ref{th:main-special-case}. 
To explain the principal construction, we introduce the notion of a tri-expander hypergraph, which  extends the conventional definition of a bipartite expander.
\begin{definition}
Let $G = (V_1, V_2, V_3, H)$ be a hypergraph  where
\begin{itemize}
\item the set of vertices consists of three disjoint parts  $V_1$, $V_2$, $V_3$ of the same cardinality
\item the set of hyperedges is a set $H\subset V_1 \times V_2 \times V_3$.
\end{itemize}
We consider three bipartite graphs $G_1$, $G_2$, $G_3$ associated with hypergraph $G$:  each  $G_i$ is a bipartite graph 
$(V_i, V_j \times V_\ell, E_i)$ (here $ j = i+1\mod 3$ and $\ell = i+2\mod 3$), where $(x, \langle y,z \rangle)\in E_i$ if and only if the triple  $\{x, y, z\}$ corresponds to a hyperedge in $H$.
The hypergraph is called \emph{tri-expander} if the graphs $G_1$, $G_2$, $G_3$ are bi-regular spectral expanders.
\end{definition}
\begin{remark}
The definition of a tri-expander and an application of the expander mixing lemma to the associated bipartite graphs (see below) seems to be similar but not literally equivalent to the definition of the second eigenvalue for $3$-uniform hypergraph and  the hypergraph generalization of the expander mixing lemma in \cite{hypergraph-mixing-lemma}.
\end{remark}

We show that the communication is costly  for a triple of inputs $(x,y,z)$ that is a hyperedge in a tri-expander.
To this end, we combine the idea from paragraph~3.3 with an argument similar to the observation sketched in the \emph{Digression}:  
each party must send a message having non-negligible mutual information with two other inputs (an information-theoretic argument)
but this is only possible when each of the messages is very long (due to the spectral bound and the expander mixing lemma).

\medskip
\noindent
\textbf{3.5. Construction of a tri-expander.}
To conclude the proof of the main result it remains to show that there exists a  tri-expander with suitable parameters:

\begin{proposition}\label{p:hypergraph}
For all integer numbers $k\ge0$ and $n\ge 1$ there exists a tri-expander $G = (V_1, V_2, V_3, H)$  such that 
\begin{itemize}
\item $\# V_1 = \# V_2 = \#V_3 = \Theta(2^{kn})$, 
\item for all $i\not= j$, for every $x\in V_i$ there exists $\Theta(2^{(k-1)n})$ vertices $y\in V_j$ such that $x$ and $y$ are adjacent in the hypergraph,
\item $\# H = \Theta(2^{kn} \cdot 2^{(k-1)n} \cdot 2^{(k-2)n})$.
\end{itemize}
\end{proposition}
\begin{proof}%[Combinatorial part of the proof.]
We construct such a tri-expander explicitly.  We fix the finite field  $\mathbb{F}_{2^n}$ with $q=2^n$ elements,
the $(k+2)$-dimensional space ${\cal L} $  over this field,
and the subspace ${\cal L}_{so} \subset {\cal L} $ that consists of self-orthogonal vectors. Observe that $\# {\cal L}_{so} = \# {\cal L} / q = q^{k+1}$ (a subspace of co-dimension $1$ in ${\cal L}$). 
Let $V$ denote the space of all \emph{directions} in  ${\cal L}_{so} $ except for the direction $(1,\ldots,1)$ (which is self-orthogonal for even $k$). Observe that $\# V =  \Theta(q^k)$.

We let  $V_1 = V_2 = V_3 = V$ and
define $H$ as the set of all triple $(x,y,z) \in V^3$ such that $x,y,z$ are \emph{distinct and pairwise orthogonal} directions in ${\cal L}_{so} $.

For every vector $x\in {\cal L}_{so}$, the condition of being orthogonal to $x$ determines in ${\cal L}_{so}$ a subspace of co-dimension $1$; this subspace consists of $q^{k}$ vectors
(including $x$ itself as it is self-orthogonal) and, respectively, $(q^{k} - 1 )/ (q-1)$ directions (again, including the direction  collinear with $x$). 
If we have two non-collinear vectors $x,y\in {\cal L}_{so}$, then the condition of being orthogonal to $x$ and $y$ determines in ${\cal L}_{so}$ a subspace of co-dimension $2$;  this subspace consists of $q^{k-1}$ vectors
(including $x$ and $y$), which corresponds to   $(q^{k-1} - 1 )/ (q-1) = \Theta(q^{k-2})$ directions (once again, including the directions collinear with $x$ and with $y$). 

Thus, we have $ \Theta(q^{k})$ individual vertices,  $\Theta(q^{k} \cdot q^{k-1})$ pairs of adjacent vertices, and $\Theta(q^{k} \cdot q^{k-1} \cdot q^{k-2})$ adjacent triples (hyperedges).
It remains to compute the eigenvalues of the associated bipartite graphs.
\begin{lemma}\label{l:tri-expander}
The hypergraph $G = (V_1,V_2,V_3,H)$ defined above is a tri-expander.
\end{lemma}
(This fact might be known, but for lack of a reference we give a proof in Appendix~\ref{sec:spectral-gap}.)
In the proof we use rich symmetries of this hypergraph.
To guarantee these symmetries, we have imposed the restrictions that may seem artificial: the characteristic of the field is $2$, we take into consideration only self-orthogonal vectors,  the direction $(1,\ldots,1)$ is excluded from $V$.
\end{proof}
\begin{remark}
A standard counting shows that for most hyperedges $(x,y,z) $ in the graph from Proposition~\ref{p:hypergraph} we have
$
\C(x) \eqp \log \Theta(q^k) \eqp kn,\ 
\C(x, y) \eqp \log \Theta(q^k \cdot q^{k-1})  \eqp (2k-1)n, \\ 
\C(x,y,z) \eqp \log  \Theta(q^k \cdot q^{k-1} \cdot q^{k-2})  \eqp (3k-2)n,
$
and we get the profile \eqref{eq:main-profile}.
\end{remark}

\section{When it is hard to say anything that the interlocutor  already knows}
\label{sec:1}

\vspace{0em}
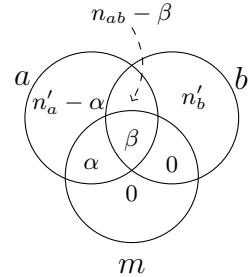
\begin{wrapfigure}{R}{0.50\textwidth}
 %   			 \begin{figure}[t]
        				 \centering
				\begin{tikzpicture}[scale=0.35]
				  \draw \firstcircle node[above left] {};%{\small $n_a'-\alpha$};
				  \draw \secondcircle node [above right] {\small $n_b'$};
				  \draw \thirdcircleup node [below] {\small $0$};
				  \node at (95:0.05)   {\small $\beta$};
				  \node at (90:1.4) {} ; %{\small $ n_{ab}-\beta$};
				  \node at (210:1.75) {\small $\alpha$};
				  \node at (330:1.75) {\small $0$};
				  \node at (150:4.75) {\large $a$};
				  \node at (30:4.75) {\large $b$};
				  \node at (270:4.75) {\large $m$};
				  
				  \node at (148:2.80) {\small $n_a'-\alpha$};
				  \node at (90: 4.80) {\small $ n_{ab}-\beta$};
				  \draw[dashed] (0,4.3)  edge[bend left,->]  (0,1.5) ;
				\end{tikzpicture}			
			\caption{The profile in Theorem~\ref{thm:1}.}
				\label{fig:profile-xyma}
%			\end{figure}
\end{wrapfigure}

In this section we explain our main technical tool. We consider randomized communication protocols with two participants, Speaker  and Listener. 
We assume that Speaker  holds an input string $a$ and Listener holds another input string $b$; we assume also that the complexity profile of the pairs $(a,b)$ is known to all parties. 
The aim of Speaker in this protocol is to send to Listener a message that  has non-negligible mutual information with Listener's input~$b$, as we discussed in Section~\ref{sec:informal-proof}.

\begin{theorem}\label{thm:1} 
Let $G=(V_L,V_R,E)$ be a bipartite spectral expander such that $N= \# V_L$, $M=\# V_R$, and $(D_L, D_R)$ are the degrees of the edges in $V_L$ and $V_R$ respectively.  
Let $(a,b) \in E$ be a ``typical'' edge in the graph, i.e., $\C(a,b) \eqp \log \#E$, and $\C(m\mid a) \eqp 0$.
Then
$
\Inf(m:b) \lep \max\{ 0, \C(m) - \C(a\mid b) \}.
$
In particular, if the length of $m$ is less than $\C(a\mid b)$, then $\Inf(m:b) \eqp 0$.
\end{theorem}
\begin{remark}\label{rem:relativization} 
The statement of Theorem~\ref{thm:1} remain valid if we relativize all terms of Kolmogorov complexity  in this statement conditional on a string $r$ such that $\Inf(r:(a,b))\eqp0$. 
In what follows we present the proof without $r$. But every step of this argument trivially relativizes conditional on $r$ assuming that $\C(a,b \mid r) \eqp \C(a,b) \eqp \log \#E$,
we only need to add routinely  the random bit string $r$ to the condition of all terms with Kolmogorov complexity appearing in the proof.  
%We will use a similar abuse of notation in other proofs as well.
\end{remark}

%This proposition claims that if Alice wants to send a message having $\delta n$ bits of mutual information with $y$, she must send a message of size at least $(1+\delta)n$. This bound is tight (up to $O(\log n)$). Indeed, w.l.o.g. we may assume that the direction $x$ is represented by the vector $(x_1,x_2,1)$. Alice can send to Bob a message that contains completely the first coordinate $x_1$ ($n$ bits of information) and a prefix of length $\delta n$ of the second coordinate $x_2$ ($\delta n$ bits of information). Given $y$ and $x_1$, Bob can fully reconstruct the direction $x$ and, in particular, the coordinate $x_2$. Therefore, $\C(m_A\mid y) \lep n$ and $\Inf(m_A:y) \gep \delta n$.

\begin{proof}[Proof of Theorem~\ref{thm:1}]
We denote $n_a:=\log N$, $n_b=\log M$, $n_a'=\log D_R$, $n_b'=\log D_L$, and $n_{ab} := n_a - n_a'$. 
Using this notation, we have
\[
\C(a) \eqp n_a,\ \C(b) \eqp n_b, \C(a\mid b) \eqp n_a', \ \C(b) \eqp n_b,\ \C(b\mid a) \eqp n_b',\ \Inf(a:b) \eqp n_{ab}. 
\]
Since Speaker computes the message $m$ given the input data $a$, we have $\C(m\mid a) \eqp 0$. 
We denote $\alpha : = \Inf(m:a\mid b)$ and $\beta := \Inf(m:a:b)$. 
It is easy to verify that $\C(m) = \alpha+\beta$.
The complexity profile for the triple $(a,b,m)$ is shown in Fig.~\ref{fig:profile-xyma}.

\smallskip

\noindent
\textbf{Case 1.} Assume that $\C(m) \le n_a' -  2\cdot \mathsf{const} \cdot  \log n$ for some $\mathsf{const}>0$ (a constant to be specified later). 
In this case, to prove the theorem,  we need to show that  $\Inf(m:y) \eqp 0$. 
In our notation this is equivalent to $\beta \eqp0$. 
More technically, we are going to show that 
\begin{equation}
\label{eq:beta}
\beta \le \mathsf{const} \cdot \log n.
\end{equation}
For the sake of contradiction we assume that \eqref{eq:beta} is false. 
It is enough to consider the case when $\beta$ is \emph{somewhat large} but not \emph{too large}, i.e., just slightly above the threshold \eqref{eq:beta}. 
Indeed, any communication protocol violating  \eqref{eq:beta} can be converted in a different protocols with  the same or a smaller value of $\alpha$ and with $\beta = \mathsf{const} \cdot  \log n + O(1)$.
To this end, we observe that by discarding a few last bits of Speaker's message $m$ we make the protocol only simpler.
So, we may replace the initial message $m$ with the shortest prefix of the initial message that still violates \eqref{eq:beta}. 
Thus, in what follows,  we assume w.l.o.g. that 
\[
\mathsf{const} \cdot  \log n < \beta \le \mathsf{const} \cdot  \log n + O(1).
\]
Let us define
$
A := \{ a' \ :\ \C(a' \mid m) \le  \C(a \mid m)  \}
\text{ and }
B := \{ b' \ :\ \C(b' \mid m) \le  \C(b \mid m)  \}.
$
We use the following standard claim:

\smallskip
\noindent
\textbf{Claim.} \emph{ $\# A =2^{\C(a \mid m) \pm O(\log n)} = 2^{n_a-\alpha-\beta \pm O(\log n)}$ and  $\# B = 2^{\C(b \mid m) \pm O(\log n)}= 2^{n_b - \beta \pm O(\log n)}$} (see, e.g.  \cite[Claim~4.7]{jacm2019}).

\smallskip

From the claim we obtain 
$
\# A  \cdot \# B = 2^{n_a- \alpha   -\beta  + n_b -\beta \pm O(\log n)} =  2^{n_a +n_b  - \C(m) - \beta  \pm O(\log n)} .
$
Since $\C(m) \le n_a' -  2\cdot \mathsf{const}  \log n$ and $\beta  <  \mathsf{const}  \log n + O(1)$, we conclude
\[
\begin{array}{rcl}
n_a + n_b  - \C(m) -\beta n \pm O(\log n) 
&\ge& n_a + n_b  - ( n_a'-  2\cdot \mathsf{const} \cdot  \log n) -  \mathsf{const} \cdot  \log n  - O(\log n) \\
&\ge& n_{ab} + n_b   +  \mathsf{const} \cdot  \log n   - O(\log n) 
\ge  n_{ab} + n_b . 
\end{array}
\]
(To get the last inequality, we should choose  the value of $  \mathsf{const} $ in \eqref{eq:beta} so that  $ \mathsf{const} \cdot \log n$  majorizes the term $O(\log n)$ in the inequality above.)
Thus,  $\# A \cdot \# B \ge 2^{n_{ab} + n_b} = \frac{M^2}{D_L}$. 

With the Expander Mixing Lemma  (Corollary~\ref{mixing-lemma-large-sets})  we obtain 
\[
E(A,B)  = O\left(\frac{D_L  \cdot \# A \cdot \# B}{M}\right) =  O\left(\frac{ \# A \cdot \# B}{M/D_L}\right).
\]
Now observe that given $m$ and the numbers $\C(a\mid m)$ and $\C(b \mid m)$ we can enumerate the sets $A$ and $B$ and, therefore, 
we can describe $(a,b)$ by the index of this edge in the list of all edges between $A$ and $B$. 
The size of such an index is $ \log E(A,B)$. 
Hence,
\[
\begin{array}{rcl}
\C(a,b \mid m) \lep \log E(A,B)  &\lep&  (n_a +n_b  - \C(m) - \beta)   - (n_b- n_b')\\
&=& n_a +n_b'  - \C(m) - \beta 
= \C(a,b)   - \C(m) - \beta,
\end{array}
\]
and 
$
\C(a,b) \lep \C(m) + \C(a,b \mid m) \lep  \C(a,b)   -\beta .
$
The terms $O(\log n)$ hidden in the notation $\lep$ and $\eqp$ in this  inequality do not depend on $\beta$.
Thus, we get a contradiction if the constant in \eqref{eq:beta} is chosen large enough.

\smallskip

\textbf{Case 2.} Now we assume that $\C(m) = n_a' + \delta $ for an arbitrary  $\delta$. 
Denote by $m'$ the prefix of $m$ of length $(n_a'-\mathsf{const} \log n)$ and by $m''$ the suffix of $m$ of length 
$(\delta  + \mathsf{const}  \log n)$.
We know from Case~1 that $\Inf(m' : b)\eqp 0$. 
It remains to apply the chain rule,
\[
\Inf(m : b) \eqp \Inf(m' : b) + \Inf(m'' : b \mid m') 
\eqp \Inf(m'' : b\mid m') \lep |m''| \eqp  \delta.
\]
and the theorem is proven.
\end{proof}
From this theorem we obtain  immediately the following corollary. 
\begin{corollary}\label{thm:1-corollary} 
Let $G=(V_L,V_R,E)$ be a bipartite spectral expander such that $N= \# V_L$, $M=\# V_R$, and $(D_L, D_R)$ are the degrees of the edges in $V_L$ and $V_R$ respectively.  

(a) We assume that  Speaker and Listener are given, respectively, $a$ and $b$ that are ends of a typical edge $(a,b) \in E$ in the graph.  
We consider a one-round communication protocol where Speaker sends to Listener a message $m = m (a)$. 
Then
$
\Inf(m:b) \lep \max\{ 0, \C(m) - \C(a\mid b) \}.
$
In particular, if the length of $m$ is less than $\C(a\mid b)$, then $\Inf(m:b) \eqp 0$.

(b) A similar statement is true if Speaker  and Listener are given instead of $a$ and $b$ some inputs $a'$ and $b'$ such that $\C(a'\mid a)\eqp 0$ and $\C(b' \mid b)\eqp 0$
(e.g., if Speaker is given a function of a vertex $a\in V_L$ and Listener is given a function of a vertex $b\in V_R$).
\end{corollary}

\section{Protocols with simultaneous  messages : a warm-up example}
\label{sec:parallel-messages}

In this section we use Theorem~\ref{thm:1}  from the previous section to prove a lower bound for communication complexity of the following problem. 
Alice and Bob hold, respectively, lines $a$ and $b$  in a plane (intersecting at one point $c$). 
They send to Charlie (in parallel, without interacting with each other) some messages so that Charlie can reconstruct the intersection point. 
We argue that the trivial protocol (where Alice and Bob send the full information on their  lines) is essentially optimal.

\begin{theorem}\label{th1}
Let Alice and Bob be given  lines in the projective plane over the finite field $\mathbb{F}_{2^n}$ (we denote them $a$ and $b$ respectively), and it is known that the lines intersect at  point $c$.
Another participant of the protocol Charlie has no input information. 
Alice and Bob (without a communication with each other)
send to Charlie messages $m_A $ and $m_B  $ so that Charlie can find $c$, see Fig.~\ref{fig:xy2z}. 
For every communication protocol for this problem, for some $a,b$ we have
$
|m_A| + |m_B| \gep 4n,
$
which means essentially that in the worst case Alice and Bob must send to Charlie all their data  \textup(for a typical pair of lines we have $  \C(a) +  \C(b) \eqp 4n$\textup).
\end{theorem}
\begin{figure}[H]
\centering
\begin{tikzpicture}[scale=0.6]
\begin{scope}
\filldraw[black] (0,2.5) circle (0pt) node[red]{\bf \large Alice};
\draw[gray] (-2,-2) -- (2,-2) -- (2,2) -- (-2,2) -- (-2,-2) ;
\draw[gray, thick,dashed] (-1,2) -- (1,-2);
\filldraw[black] (1.4,1.6) circle (0pt) node{$a$};
\draw[red, thick] (-2,-1.8) -- (2,1.8);
\end{scope}

\begin{scope}[shift={(0,-7)}]
\filldraw[black] (0,2.5) circle (0pt) node[blue]{\bf \large Bob};
\draw[gray] (-2,-2) -- (2,-2) -- (2,2) -- (-2,2) -- (-2,-2) ;
\draw[blue, thick] (-1,2) -- (1,-2);
\draw[gray, thick,dashed] (-2,-1.8) -- (2,1.8);
\filldraw[black] (1.0,-1.4) circle (0pt) node{$b$};
 \end{scope}

\begin{scope}[shift={(14,-3)}]
\filldraw[black] (-3.3,0) circle (0pt) node{\bf \large Charlie};
\draw[gray] (-2,-2) -- (2,-2) -- (2,2) -- (-2,2) -- (-2,-2) ;
\draw[gray, thick,dashed] (-1,2) -- (1,-2);
\draw[gray, thick,dashed] (-2,-1.8) -- (2,1.8);
\filldraw[black] (0,0) circle (2pt) node[anchor=west]{$c= \,?$};
 \end{scope}

\draw [->,ultra thick,red] (2.5,-0.0) to [out=10, in=120]  (10,-2.6);
\filldraw[black] (3.5,0.7) circle (0pt) node[anchor=west]{\large $m_A=m_A(a,\text{public random bits})$};

\draw [->,ultra thick,blue] (2.5,-6.5) to  [out=-10, in=-120]   (10,-3.5);
\filldraw[black] (3.5,-7.2) circle (0pt) node[anchor=west]{\large $m_B=m_B(b,\text{public random bits})$};

\end{tikzpicture}
\caption{Alice holding $a$ and Bob holding $b$ send simultaneous messages to Charlie, who computes~$c$.}
\label{fig:xy2z}
\end{figure}
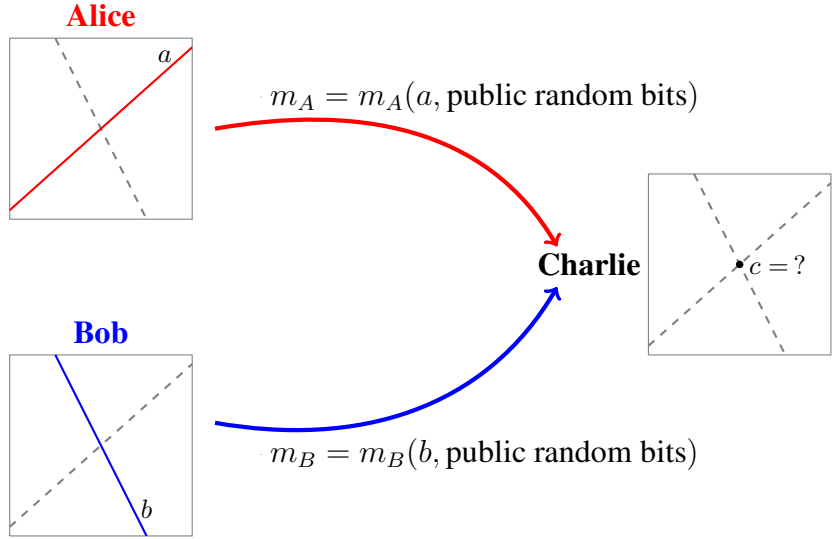
In the setting of Theorem~\ref{th1}, the inputs of Alice and Bob contain $n$ bits of the mutual information with $c$, so an easy lower bound for the communication complexity is $n+n=2n$,
see Fig.~\ref{fig:two-lines-and-point}.
However, due to the spectral properties of graphs implicitly present in this construction, the true communication complexity of this problem is twice bigger. 
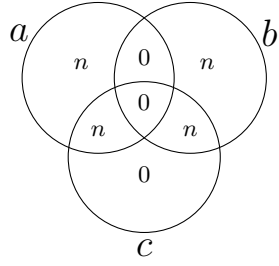
\begin{figure}[h]
\centering
				\begin{tikzpicture}[scale=0.4]
				  \draw \firstcircle node[above left] {\small $n$};
				  \draw \secondcircle node [above right] {\small $n$};
				  \draw \thirdcircle node [below] {\small $0$};
				  \node at (95:0.05)   {\small $0$};
				  \node at (90:1.55) {\small $ 0$};
				  \node at (210:1.75) {\small $n$};
				  \node at (330:1.75) {\small $n$};
				  \node at (150:4.75) {\Large $a$};
				  \node at (30:4.75) {\Large $b$};
				  \node at (270:4.75) {\Large $c$};
				\end{tikzpicture}
\caption{Complexity profile for two lines ($a$ and $b$) and their intersection point $c$ in the plane over $\mathbb{F}_{2^n}$.}
\label{fig:two-lines-and-point}
\end{figure}
\begin{proof}[Sketch of the proof.]
In this sketch we ignore the public randomness and explain the argument for deterministic protocols. 
A generalization for protocols with public randomness  is pretty straightforward, see the full proof below.

Let $(a,b)$ be a pair of lines in a  projective  plane over $\mathbb{F}_{2^n}$  intersecting at a point $c$,  such that 
$
\C(a,b) \eqp \C(a) + \C(b) \eqp 4n
$
(which is the case for most pairs of lines in the plane).   
Observe that $\Inf(a:c)\eqp n$ and $\Inf(b:c)\eqp n$. 
It follows that for the messages $m_A= m_A(a)$ and $m_B=m_B(b)$ we have $\Inf(m_A:c)\lep n$ and $\Inf(m_B:c)\lep n$.
Using standard information theoretic inequalities, one can show that Alice's message $m_A$ and Bob's message $m_B$ determine the point $c$ uniquely only if 
 $\Inf(m_A  : c)  \eqp  n$ and $\Inf(m_B : c) \eqp n$. Thus, Alice and Bob must send messages with  large enough information on $c$.

The graph of possible pairs $(a,c)$ and the graph of possible pairs $(b,c)$ (the configurations $(\text{line}, \text{point})$)
is the same as  in Example~\ref{example:line-point}. 
Hence, we can apply Theorem~\ref{thm:1} (Alice and Bob play the roles of Speaker, and Charlie plays the role of Listener) and conclude that
$
\Inf(m_A:c) \lep \max\{ 0, \C(m_A) - n\}
\text{ and }
\Inf(m_B:c) \lep \max\{ 0, \C(m_B) - n\}.
$
In particular,  $\Inf(m_A  : c)  \eqp  n$ and $\Inf(m_B : c) \eqp n$ only if Kolmogorov complexities of $m_A$ and $m_B$ are both at least $2n$. Thus, the total communication complexity is $\gep 2n+2n=4n$.
\end{proof}
In what follows we present the full proof of Theorem~\ref{th1}. The reader can skip it and proceed to the proof of the main result in the next section.

\begin{proof}[Full proof of Proof of Theorem~\ref{th1}.]
Let $(a,b)$ be a pair of lines in a  projective  plane over $\mathbb{F}_{2^n}$ such that 
\[
\C(a,b) \eqp \C(a) + \C(b) \eqp 4n
\]
(which is the case for most pairs of lines in the plane), and let  $c$ be the point of intersection of these lines, see Fig.~\ref{fig:two-lines-and-point}.

Denote by $r$ the string of random bits from the public source of randomness (accessible to Alice, Bob, Charlie, and to the eavesdropper).
We assume $r$ and the inputs $(a,b)$ are independent, i.e., $\Inf(r:a,b)\eqp0$ (this is the case with an overwhelming probability). 
For such a string  $r$ all terms with Kolmogorov complexities involving $a,b,c$ do not change if we add $r$ in the condition: 
\[
\begin{array}{l}
\C(a,b,c \mid r) \eqp \C(a,b,c) \eqp \C(a,b), \  \C(a,c\mid r) \eqp \C(a,c),\  \C(b,c\mid r) = \C(b,c),\\  
\C(a\mid r) = \C(a),\ \C(b\mid r) = \C(b),\ \C(c\mid r) = \C(c).
\end{array}
\]
This implies 
\[
\C(a \mid c,r) \eqp \C(a , c \mid r) - \C(c\mid r) \eqp \C(a,c) - \C(c) \eqp \C(a\mid c), 
\]
and, similarly, 
$\C(b \mid c,r)   \eqp \C(b \mid c)$ and $\C(a,b \mid c,r)  \eqp \C(a,b \mid c)$.

Observe that the graph of possible pairs $(a,c)$ and the graph of possible pairs $(b,c)$ (the configurations $(\text{line}, \text{point})$ on the projective plane) is the same as  in Example~\ref{example:line-point}. 
Hence, we can apply Theorem~\ref{thm:1} (Alice and Bob play the roles of Speaker, and Charlie plays the role of Listener;  all terms are relativized conditional on $r$, see Remark~\ref{rem:relativization}) and conclude that
\begin{equation}\label{eq:3}
\Inf(m_A:c \mid r) \lep \max\{ 0, \C(m_A \mid r) - n\}
\text{ and }
\Inf(m_B:c\mid r) \lep \max\{ 0, \C(m_B\mid r) - n\}.
\end{equation}
In particular,  $\Inf(m_A  : c\mid r)  \gep  n$ and $\Inf(m_B : c\mid r) \gep n$ only if Kolmogorov complexities of $m_A$ and $m_B$ are both at least $2n$.

It is easy to verify that $\Inf(a:c\mid r ) \eqp \Inf(a:c) \eqp n$ and $\Inf(b:c\mid r) \eqp \Inf(b:c)\eqp n$. 
Since $m_A$ and $m_B$ are computed from $(a,r)$ and $(b,r)$ respectively, we conclude that 
 \begin{equation}\label{eq:mambz}
 \Inf(m_A:c\mid r) \lep n \text{ and }\Inf(m_B:c\mid r)\lep n. 
 \end{equation}
 We need to show that these two inequalities turn into equalities.
Indeed, by construction,
\[
\begin{array}{rcl}
\Inf(a:b\mid c,r)  &\eqp& \C(a \mid c,r) + \C(b \mid c,r)  - \C(a,b \mid c,r) \\
&\eqp& \C(a \mid c) + \C(b \mid c)  - \C(a,b \mid c) \\
&\lep& n + n - \C(a,b \mid c)  \text{ [we need $n+O(1)$ bits to specify a line given a point] } \\
&\lep& 2n -(\C(a,b) - \C(c))  \\
&\lep& 2n - 4n + 2n \eqp 0.  \text{ [we need $2n+O(1)$ bits to specify a point in the plane] }
\end{array}
\]
As $\C(m_A \mid a,r) \eqp 0$ and $\C(m_B \mid b,r) \eqp 0$, we have  (see Lemma~\ref{lemma:appendix1}(iii))
\[
\Inf(m_A:m_B\mid c,r)  \lep \Inf(a:b\mid c,r) \eqp0.
\]
Therefore, $\Inf(m_A:m_B:c \mid r)  \eqp \Inf(m_A:m_B \mid r) - \Inf(m_A:m_B\mid c , r) \gep0$, and
\[
\Inf(m_A m_B : c\mid r) \eqp \Inf(m_A : c\mid r) + \Inf( m_B : c\mid r) - \Inf(m_A:m_B:c\mid r) \lep  \Inf(m_A : c\mid r) + \Inf( m_B : c\mid r).
\]
On the other hand, since Charlie can compute $c$ given the messages $(m_A, m_B)$ and the string of random bits $r$, we have 
\[
\Inf(m_A m_B : c \mid r) \eqp \C(c\mid r) - \C(c\mid m_A m_B ,r)  \eqp 2n - 0  \gep 2n,
\] 
and, therefore,
\[
  \Inf(m_A : c\mid r) + \Inf( m_B : c\mid r) \eqp 2n.
\]
Keeping in mind \eqref{eq:mambz},
we conclude that  $\Inf(m_A : c\mid r) \eqp n$ and $\Inf( m_B : c\mid r)\eqp n$. 
Due to \eqref{eq:3}, this is possible only of  $\C(m_A)\gep \C(m_A \mid r) \gep 2n$ and $\C(m_B)\gep \C(m_B\mid r) \gep 2n$.
This means that the total length of the sent messages is at least $2n+2n = 4n$ bits.
\end{proof}

\section{Secret key agreement: a lower bound for  the most crucial profile} 
\label{sec:crypto-special-case}

In this section   we prove a lower bound for communication complexity of secret key agreement with three parties. 
Let us recall the setting.
We assume that Alice, Bob, and Charlie are given inputs $x$, $y$,  $z$ respectively with the complexity profile \eqref{eq:main-profile}.
as shown in Fig.~\ref{fig:profiles}(b).
This is a pretty ``generic'' complexity profile;
by choosing $k$, we  control the gap between the complexities of $x,y,z$ and the mutual informations shared by the inputs.

We consider communication protocols with public randomness. 
Denote by $r$ the string of random bits accessible for all the parties (including the eavesdropper).
We assume that Alice, Bob, and Charlie  broadcast simultaneously messages
$
m_A = m_A(x,r),\
m_B = m_A(y,r),\
m_C = m_A(z,r)
$
over a public communication channel. 
Then each of them computes the final result 
\[
\begin{array}{l}
\text{key}_{\rm Alice} (x,r,m_B, m_C), \
\text{key}_{\rm Bob} (y,r,m_A, m_C), \
\text{key}_{\rm Charlie} (z,r,m_A, m_B).
\end{array}
\]
We say that a protocol  is successful if
$
\text{key}_{\rm Alice} = \text{key}_{\rm Bob} = \text{key}_{\rm Charlie} = w
$
(i.e., the parties agree on a common key $w$) and $\C(w  \mid   \langle m_A, m_B, m_C,r\rangle) \eqp |w|$ (i.e., the eavesdropper gets no information on this key).

Theorem~\ref{th:algorithmic-3} claims that for any $\epsilon>0$ there exists a protocol that is successful with probability $(1-\epsilon)$, and the size of the key is equal to \eqref{eq:key-size},
which gives  for the profile \eqref{eq:main-profile} the value $1.5n$.
Moreover, this value of the key is optimal  (up to an additive term $O(\log n)$). 

It was shown in  \cite{jacm2019} that  a secret key of this size can be obtained in an \emph{omniscience} protocol.
In this protocol, the parties broadcast messages so that each of them learns completely the entire triple of inputs $(x,y,z)$.  
The total length of the broadcasted messages bits is less than $\C(x,y,z)$, so an eavesdropper can learn only a partial information on the inputs. 
More specifically, communication complexity of the omniscience protocol  is \eqref{eq-cc-3},
which is $(3k-4.5)n$ for a triple satisfying \eqref{eq:main-profile}.
The gap between $\C(x,y,z) \eqp (3k-3)n$ and the amount of the divulged information is used to produce the secret key of size $1.5n$.

The omniscience protocol used in \cite{jacm2019} provides an \emph{upper bound} on the communication complexity of secret key agreement. 
In what follows we prove the matching  \emph{lower bound}  (for protocols with simultaneous messages) and show that $(3k-4.5)n$  is the optimal  communication complexity
for a protocol of secret key agreement protocols with simultaneous messages for  inputs satisfying \eqref{eq:main-profile}.
The proof follows the scheme sketched in Section~\ref{sec:informal-proof}.
The first ingredient of this proof is Lemma~\ref{lemma:1} (see p.~\pageref{sec:informal-proof}).
\begin{proof}[Sketch of proof of Lemma~\ref{lemma:1}]
This lemma is a relativized version of   \cite[Theorem~4.2]{jacm2019}, where $s$  is used as an oracle. 
One can follow the argument from  \cite{jacm2019} step by step, substituting $s$ as a supplementary   condition in each  term of Kolmogorov complexity appearing in the proof. 
 %An alternative way to explain the argument is to define a version of Kolmogorov complexity for Turing machines  that can access $s$ as on the oracle tape and observe that the proof of Theorem~4.2 in \cite{jacm2019} relativizes.
\end{proof}

\begin{corollary}\label{col:1}
Consider a communication protocol with three parties  where Alice is given $x$, Bob is given $y$, and Charlie is given $z$. 
Denote by $m_C$ the concatenation of all messages broadcasted  by Charlie during the communication.
If the parties agree on a secret key $w$  on which  the eavesdropper gets no information (even given access to the messages sent by all parties), then 
$
\C(w) \lep \Inf(x:y \mid r, m_C).
$
\end{corollary}
\begin{proof}
We apply Lemma~\ref{lemma:1}  substituting $m_C$ instead of the public information $s$.
\end{proof}

Now we are ready to prove our main result.

\smallskip
\noindent
\textbf{Theorem~\ref{th:main-special-case} rephrased.} {\it 
Let  Alice, Bob, and Charlie be given $x$,  $y$, and  $z$ respectively such that $(x, y, z)$ is a hyperedge of the hypergraph 
$G = (V_1,V_2,V_3, H)$
from Proposition~\ref{p:hypergraph}
(the pairwise disjoint self-orthogonal directions in a $(k+2)$-dimensional vector space over $\mathbb{F}_{2^n}$).
We consider non-interactive communication protocols where Alice, Bob, and Charlie send messages $m_A $, $m_B $, and $m_C $ respectively
and produce a secret key $w$ with the optimal complexity $\C(w)\eqp 1.5n$. 
Then
$
\C(m_A) \gep (k-1.5)n, \ \C(m_B) \gep (k-1.5)n,\ \C(m_C) \gep (k-1.5)n,
$
and the communication complexity of the protocol is at least  $(3k-4.5)n - O(\log n)$,
which matches the communication complexity of the omniscience protocol.
}
\begin{proof}
To simplify the notation, we ignore the bits $r$ provided by the public source of randomness and  explain the proof for  deterministic protocols. 
Our argument trivially relativizes given any instance of random bits $r$ independent of $(x,y)$ (which is true with a probability close to $1$), 
cf. the full proof of Theorem~\ref{th1}.

From Corollary~\ref{col:1} we know that the size of the key (in our case $1.5n$) cannot be greater than $\Inf(x:y\mid m_C)$.  
By the construction of the tri-expander, $\Inf(x:y) \eqp n$.  
Therefore, the difference between $\Inf(x:y)$ and  $\Inf(x:y\mid m_C)$ is at least $0.5n$. 
\begin{lemma}\label{lemma:ineq2}
For all  binary strings $x,y,z$ it holds 
$
\Inf(x:y\mid s)  - \Inf(x:y) \lep \Inf(s:xy).
$
\end{lemma}
\noindent
(See the proof of the lemma in Appendix~\ref{sec:inequalities}.)
We combine Corollary~\ref{col:1} with Lemma~\ref{lemma:ineq2} and obtain $\Inf(m_C : xy) \gep 0.5n$.

Now we apply Theorem~\ref{thm:1} to the bipartite graph $G_3$ associated with the tri-expander $G$  (see p.~\pageref{p:hypergraph}); here Charlie plays the role of Speaker, and Alice and Bob together play the role of Listener.
Since $\Inf(m_C : xy) \gep 0.5n$, we obtain 
$
\C(m_C) \gep \C(z \mid x,y) + 0.5n \eqp (k-1.5)kn.
$
A~similar argument applies to  $\C(m_A)$ and $\C(m_B) $,  and we are done.
\end{proof}

\section{Secret key agreement: a lower bound for all symmetric  profiles}
\label{sec:crypto-general-case}

\begin{proof}[Proof of Theorem~\ref{th:main}]
If the complexity profile of  $(x,y,z)$ is symmetric then it can be specified by a triple of parameters $\alpha,\beta,\gamma$,
\begin{equation}\label{eq:sym-profile-greek-letters}
\left\{
\begin{array}{l}
\C(x\mid y,z) \eqp \C(y\mid x,z) \eqp \C(z\mid x,y) \eqp \alpha ,\\
\Inf(x:y\mid z) \eqp \Inf(x:z\mid y) \eqp \Inf(y:z\mid x) \eqp \beta ,\ 
\Inf(x:y:z) \eqp \gamma .
\end{array}
\right.
\end{equation} 
In Theorem~\ref{th:main-special-case} we proved that communication complexity \eqref{eq-cc-3} of the omniscience protocol is optimal in case 
$\alpha = (k-2)n$,  $\beta = n$, and $\gamma = 0$.  
We  reduce the problem with  arbitrary $\alpha,\beta,\gamma$ to the special case settled in  Theorem~\ref{th:main-special-case}. 
We split this reduction into three steps, as shown in the following lemma.
\begin{lemma} 
\label{l:reduction}
If  communication complexity \eqref{eq-cc-3}  is optimal (in the worst case) for some triples of inputs $(x,y,z)$ with complexity profile \eqref{eq:sym-profile}
then 

\begin{itemize}
\item
(a) for every positive $\delta\le n $, communication of the omniscience protocol is also optimal (also in the worst case) for some triples of inputs $(x',y',z')$ with complexity profile
\begin{equation}\label{eq:reduction-a}
\left\{
\begin{array}{l}
\C(x'\mid y',z') \eqp \C(y'\mid x',z') \eqp \C(z'\mid x',y') \eqp \alpha-\delta ,\\
\Inf(x':y'\mid z') \eqp \Inf(x':z'\mid y') \eqp \Inf(y':z'\mid x') \eqp \beta ,\ 
\Inf('x:y':z) \eqp \gamma,
\end{array}
\right.
\end{equation}
\item
(b) for every positive $\delta$, communication of the omniscience protocol is also optimal for some triples of inputs $(x',y',z')$ with complexity profile
\begin{equation}\label{eq:reduction-b}
\left\{
\begin{array}{l}
\C(x'\mid y',z') \eqp \C(y'\mid x',z') \eqp \C(z'\mid x',y') \eqp \alpha ,\\
\Inf(x':y'\mid z') \eqp \Inf(x':z'\mid y') \eqp \Inf(y':z'\mid x') \eqp \beta ,\ 
\Inf(x':y':z) \eqp \gamma+\delta.
\end{array}
\right.
\end{equation}
\end{itemize}
\vspace{-0.5em}
Let us consider the special case $\alpha \eqp (k-2)n$, $\beta \eqp n$, $\gamma\eqp0$
(as in Theorem~\ref{th:main-special-case}). Then 
\begin{itemize}
\item 
(c)   for every positive $\delta\le \beta/2 $, communication of the omniscience protocol is also optimal for some triples of inputs $(x',y',z')$ with complexity profile
\begin{equation}\label{eq:reduction-c}
\left\{
\begin{array}{l}
\C(x'\mid y',z') \eqp \C(y'\mid x',z') \eqp \C(z'\mid x',y') \eqp \alpha ,\\
\Inf(x':y'\mid z') \eqp \Inf(x':z'\mid y') \eqp \Inf(y':z'\mid x') \eqp \beta +\delta,\ 
\Inf('x:y':z) \eqp  -3\delta ;
\end{array}
\right.
\end{equation}
\vspace{-0.5em}
\end{itemize}
\end{lemma}
In Lemma~\ref{l:reduction} we show that the existence of a ``too efficient'' protocol for \eqref{eq:reduction-b}, \eqref{eq:reduction-c}, \eqref{eq:reduction-a} would imply a ``too efficient protocol'' for \eqref{eq:main-profile}, which is impossible due to Theorem~\ref{th:main-special-case}. 
The proof  is based on repeated application of Muchnik's theorem on conditional descriptions (\cite{muchnik-coding}), 
which basically claims that for all strings $a,b_1,\ldots,b_\ell$ and for every number $m\le \C(a)$ 
there exists a ``digital fingerprint'' of $a$ of length $m$ that looks maximally random conditional on each $b_j$. 
Technically, this means that for some $a'$ we have
\[
\C(a') \eqp m,\ \C(a' \mid a) \eqp 0, \ \text{and} \ \C(a' \mid b_j) = \min\{ \C(a'\mid b_j), m\} \ \text{for}\ j=1,\ldots,\ell.
\]
The proof of this lemma uses mostly techniques of Kolmogorov complexity that are not specific for communication problems, 
see Appendix~\ref{sec:reduction}.

It is not hard to verify that starting with a triple $(x,y,z)$ from Theorem~\ref{th:main-special-case} and then applying the reductions from Lemma~\ref{l:reduction}, we can obtain any realizable profiles \eqref{eq:sym-profile}. 
Indeed, we begin with a triple of pairwise orthogonal directions $(x,y,z)$ with $\alpha = (k-2)n, \beta = n, \gamma=0$ for a suitable $n$ and  $k$, then apply Lemma~\ref{l:reduction}~(b) or Lemma~\ref{l:reduction}~(c) to get a triple $(x',y',z')$ with a suitable  $I(x':y':z')$ (case~(b) serves to make the triple mutual information positive, and case~(c) is needed if we want to make it negative), and further   apply Lemma~\ref{l:reduction}~(a) to trim the value of $\alpha$. 

Thus, Theorem~\ref{th:main-special-case} implies optimality of  \eqref{eq-cc-3} 
not only for triples with a pretty specific complexity profile but 
for triples of inputs  $(x,y,z)$ with arbitrary symmetric complexity profile  \eqref{eq:sym-profile}. 
\end{proof}

\section{Upper bound for interactive protocols}\label{sec:interactive}

In this section we show that the communication complexity  \eqref{eq-cc-3}  is not optimal for multi-round protocols where the parties can actually interact with each other.

\begin{proposition}\label{thm:5}
In the setting of Theorem~\ref{th:main-special-case} there is a \emph{multi-round}  communication protocol  \textup(not a \emph{simultaneous messages} protocol\textup) with communication complexity 
$
(2k-2.5)n + O(\log n),
$
where the parties agree on a secret key of the optimal size 
$1.5n-O(\log n)$.
\label{th:upper-bound}
\end{proposition}

\begin{proof}[Sketch of  proof of Proposition~\ref{thm:5}.]
We adapt the \emph{omniscience} protocol from \cite{jacm2019}. 
In what follows we assume that random hash-functions are chosen with the public source of randomness
(e.g., one may assume that random hashing is the multiplication by a randomly chosen matrix).

In the first round, Alice and Bob send messages $m_A= m_A(x,r)$ and  $m_B= m_B(y,r)$
(random hash-values of $x$ and $y$), each of length $(k-1.5)n+ O(\log n)$
such that Charlie given $(m_A, m_B, z)$ can reconstruct the pair $(x,y)$. 
Then Charlie sends a message $m_C$ that is another random hash-value of $(x,y)$
of length $0.5 n+ O(\log n)$.

With a high probability (for a randomly chosen hash-function), 
the values  $m_B$ and $m_C$ are enough for Alice to reconstruct $y$, and 
the values  $m_A$ and $m_C$ are enough for Bob to reconstruct $x$. 
Thus, at the end of communication, with high probability each party knows $(x,y)$. 
At the same time, the adversary learns from the communication at most $|m_A| + |m_B| + |m_C| = (2k-2.5)n + O(\log n)$ bits of information.

Now each party applies to $(x,y)$ another (independently chosen) random hash function and obtains a hash-value $w=\mathrm{hash}(x,y)$ of length 
\[
\C(x,y) - |m_A| - |m_B| - |m_C| - O(\log n) = 1.5n - O(\log n).
\]
With a high probability the obtained $w$ is incompressible conditional on the data accessible to the eavesdropper (the messages of the parties and the public random bits).
\end{proof}
\begin{remark}
In the omniscience protocol, 
we may define random hashing as random linear mappings, i.e., each hash-values can be computed as the product over $\mathbb{F}_2$ of a bit vector by a randomly chosen binary matrix  of the appropriate dimension. 
These matrices can be  made publicly known: we can obtain these random matrices from the public source of random bits, 
and this does not reveal any information about the secret key to the adversary.
Using more sophisticated  constructions of hash-functions, we could reduce the number of used random bits 
(although this improvement is not necessary to prove Theorem~\ref{thm:5} in the model with a public source of randomness).
\end{remark}

\section{Conclusion and open problems}

We proved that the standard omniscience protocol provides the optimal worst-case communication complexity of the problem of secret key agreement  (with three parties) in the class of protocols with simultaneous messages.
A general open problem is to study the limits of our approach.
In particular, for the class of multi-round communication protocols,  the value \eqref{eq-cc-3} is no longer the optimal communication complexity of secret key agreement (Theorem~\ref{th:upper-bound}).
Our technique implies \emph{some} lower bounds for communication complexity of interactive protocols, but it does not match the known upper bounds. 
Thus, a natural open problem is to settle the communication complexity of multi-party secret key agreement  for  multi-round protocols.
It would be also interesting to extend our results to the communication model with private sources of randomness.
Another open problem is to get rid of Lemma~\ref{l:reduction} and find a more direct proof of Theorem~\ref{th:main} with a more flexible construction of a tri-expander.

\appendix

\section{Preliminaries: Kolmogorov complexity in some more detail}
\label{sec:prelim}

Let $M$ be a Turing Machine with two input tapes and one output tape. 
We say that $p$ is a program that prints a string $x$ given $y$ (a description of $x$ conditional on $y$)  if $M$ prints $x$ on the pair of inputs $(p, y)$.  
\emph{Kolmogorov complexity} of $x$ \emph{conditional on} $y$ relative to $M$ is defined as 
\[
\C_M(x\mid y)=\min\{|p|:M(p,y)=x\}.
\]
The invariance theorem  (see \cite{kolmogorov}) claims that there exists an \emph{optimal} Turing machine $U$ such that for every other Turing machine $V$ there is a number $c_V$ such that for all $x$ and $y$
\[
\C_U(x\mid y)\le \C_V(x\mid y)+c_V.
\]
Thus, the algorithmic complexity of $x$ relative to $U$ is minimal up to an additive constant. 
In the rest of the paper we fix an optimal machine $U$, omit the subscript $U$ and define  \emph{Kolmogorov complexity} of $x$ conditional on $y$ as
\[
\C(x\mid y):=\C_U(x\mid y).
\]
Kolmogorov complexity $\C(x)$ of a string $x$ (without a condition) is defined as the Kolmogorov complexity of $x$ conditional on the empty string. 
We fix an arbitrary computable bijection between binary strings and all finite tuples of binary strings and define  Kolmogorov complexity of a tuple $\langle x_1,\ldots,x_k\rangle$ as Kolmogorov complexity of the code of this tuple. 
For brevity we denote this complexity by $\C(x_1,\ldots,x_k)$. 
Similarly, we can fix a bijection between binary strings and elements of  finite fields, polynomials over finite fields,  directions in vector  spaces over finite fields, etc., and talk about Kolmogorov complexities of these objects (implying Kolmogorov complexity of their codes).
We use the conventional notation
\[
\Inf(x:y) := \C(x) +\C(y) - \C (x,y)
\text{ and }
\Inf(x:y \mid z) := \C(x\mid z) + \C(y\mid z) - \C (x,y\mid z)
\]
(mutual  information and conditional mutual information for a pair) and % the somewhat less standard notation
\[
\Inf(x:y:z) := \Inf(x:y) - \Inf(x:y\mid z) 
\]
(the triple mutual information). 
The Kolmogorov--Levin theorem, \cite{zvonkin-levin}, claims that for  all $x,y$
\[
\C(x,y) \eqp \C(x\mid y) + \C(y).
\]
 Using the Kolmogorov--Levin theorem it is not hard  to show that
\[
 \Inf(x:y:z)  \eqp  \C(x) + \C(y) + \C(z) - \C(x,y) - \C(x,z) - \C(y,z) + \C(x,y,z)
% \eqp  \Inf(x:y) - \Inf(x:y\mid z)
%\eqp  \Inf(x:z) - \Inf(x:z \mid y) 
 %\eqp  \Inf(y:z) - \Inf(y:z\mid x).
\]
and, therefore, 
\[
 \Inf(x:y:z) \eqp  \Inf(x:z) - \Inf(x:z\mid y)  \eqp  \Inf(y:z) - \Inf(y:z\mid x). 
\]
These relations can be observed on a Venn-like diagram, see  Fig.~\ref{x-y-z}.

\smallskip

A string $x$ is said to be (almost) \emph{incompressible} given $y$ if
$
C(x\mid y) \gep |x| ,
$
and $x$ and $y$ are said to be \emph{independent}, if $\Inf(x:y)\eqp0$. 
For every $n$, the majority of  binary strings of length $n$ are almost incompressible; the vast majority of pairs of strings $x$ and $y$ of length $n$ are independent.

\smallskip

			\begin{figure}[t]
				\centering
				\begin{tikzpicture}[scale=0.85]
				  \draw \firstcircle node[above left] {\small $\C(x|y,z)$};
				  \draw \secondcircle node [above right] {\small $\C(y|x,z)$};
				  \draw \thirdcircle node [below] {\small $\C(z|x,y)$};
				  \node at (95:0.30)   {\small $\Inf(x:y:z)$};
				  \node at (90:1.55) {\small $ \Inf(x:y|z)$};
				  \node at (210:1.75) {\small $\Inf(x:z|y)$};
				  \node at (330:1.75) {\small $\Inf(y:z|x)$};
				  \node at (150:4.75) {\Large $x$};
				  \node at (30:4.75) {\Large $y$};
				  \node at (270:4.75) {\Large $z$};
				\end{tikzpicture}
				\caption{Complexity profile for a triple $x,y,z$. 
				On this diagram it is easy to observe several standard equations: 
				\newline
				$\bullet$ $\C(x) \eqp \C(x\mid y,z) +  \Inf(x:y\mid z) +  \Inf(x:z\mid y) + \Inf(x:y:z)$ 
				%\newline
				%\rule{3mm}{0mm}(the sum of all quantities inside the left circle representing $x$);
				%
				%\newline
				%\rule{3mm}{0mm}(the sum of the quantities in the intersection of the left and the right circles
				%\rule{3mm}{0mm}representing $x$ and $y$ respectively);
				%
				\newline
				$\bullet$ $\C(x,y) \eqp \C(x\mid y,z) + \C(y\mid x,z) +  \Inf(x:y\mid z) +  \Inf(x:z\mid y)  +  \Inf(y:z\mid x) + \Inf(x:y:z)$ 
				%\newline
				%\rule{3mm}{0mm}(the sum of all quantities inside the union of the left and the right circles);
				%				
				\newline
				$\bullet$ $\C(x\mid y) \eqp \C(x\mid y,z) +  \Inf(x:z\mid y) $ 
				%\newline
				%\rule{3mm}{0mm}(the sum of the quantities inside the left circle but outside the right one);
				\newline
				$\bullet$ $\Inf(x:y) \eqp \Inf(x:y\mid z) + \Inf(x:y:z)$ 
				\newline
				$\bullet$ $\Inf(x:yz) \eqp \Inf(x:y\mid z) + \Inf(x:z\mid y) + \Inf(x:y:z)$ 
				\newline
				and so on; all these equations are valid up to $O(\log(|x|+|y|+|z|))$.
				}\label{x-y-z}
			\end{figure}

For a pair of strings $(x,y)$ we call by its \emph{complexity profile} the triple of numbers $(\C(x), \C(y), \C(x,y))$. 
Due to the Kolmogorov--Levin theorem, the complexity profile of a pair is determined
(up to additive error terms $O(\log (|x|+|y|))$)  by the triple of numbers $(\C(x\mid y), \C(y\mid x), \Inf(x:y))$. 
Indeed, 
\[
\C(x) \eqp \C(x\mid y) + \Inf(x:y),\ 
C(y) \eqp \C(y\mid x) + \Inf(x:y),\
C(x,y) \eqp \C(x\mid y) + \C(y\mid x) + \Inf(x:y).
\]
Similarly, for a triple of strings $(x,y,z)$ we define its  \emph{complexity profile} as the vector with $7$ components
\[
 (\C(x), \C(y), \C(z), \C(x,y), \C(x,z), \C(y,z), \C(x,y,z)  ).
\]
This profile can be equivalently specified (again, up to  additive logarithmic  error terms) by the numbers 
\[
\C(x\mid y,z), \C(y\mid x,z), \C(z\mid x,y), \Inf(x:y\mid z),  \Inf(x:z\mid y),  \Inf(y:z\mid x),  \Inf(x:y:z),
\]
see   Fig.~\ref{x-y-z}. 

In general, for an $n$-tuple of string $(x_1,\ldots, x_n)$, its complexity profile is the vector that consists of $2^n-1$ components 
$\C(x_{i_1},\ldots, x_{i_s})$ for all non-empty tuples $1\le i_1<\ldots < i_s\le n$.
In Fig.~\ref{fig:profiles} we show diagrams illustrating the complexity profiles for Examples~\ref{example:1}-\ref{example:line-point} and Proposition~\ref{p:hypergraph}

For a survey of the basic  properties of Kolmogorov complexity we refer the reader to the introductory chapters in \cite{li-vitanyi} and \cite{shen-vereshchagin}.

\section{Bound of the spectral gap for the tri-expander}\label{sec:spectral-gap}

In the proof of  Lemma~\ref{l:tri-expander} we use the fact that the bipartite graphs associated with the hypergraph from Proposition~\ref{p:hypergraph} are highly symmetric, 
and these symmetries  simplify the computation of the eigenvalues.
To guarantee this property, we have imposed restrictions that may seem artificial: the characteristic of the field is $2$, we take into consideration only self-orthogonal vectors, and the direction $(1,\ldots,1)$ is not included  in the set $V$.

\begin{proof}[Proof of Lemma~\ref{l:tri-expander}]
By construction, the graphs $G_1$, $G_2$, $G_3$ are isomorphic. 
So  we only need to compute the eigenvalues of $G_1$.
Let us begin with the case when $k$ is an odd number, and the vector $(1,1,\ldots,1)$ is not self-orthogonal.

In the bipartite graph $G_1 = (L,R,E)$ the left part of vertices $L$ coincides with the space of all self-orthogonal directions $V$, the right part $R$ consists of pairs of self-orthogonal directions  
$(y,z)$ that are mutually orthogonal and $y\not=z$,
and the set of hyperedges $H$ consists of the triples $(x,y,z)$ of self-orthogonal directions that are pairwise distinct and mutually orthogonal.

To compute $\#L $, we count the number of directions in ${\cal L}_{so}$. To this end,  we divide the total number of non-zero vectors in the $(k+1)$-dimensional space ${\cal L}_{so}$ by the number of vectors in each equivalence class (the number of non-zero elements in the field), which gives $\# L = \Theta(2^{kn})$.

For each $y\not=(1,\ldots,1)$, the space of vectors  $z\in {\cal L}_{so}$ that are orthogonal to $y$ is a subspace of co-dimension $1$ in ${\cal L}_{so}$. 
To count the number of directions in this subspace, we need again divide the total number of non-zero vectors  by the number of non-zero elements in the field.
We obtain that $\#R = \Theta(2^{kn}) \cdot \Theta(2^{(k-1)n})$.

To find $D_R$,  we count the number of directions in ${\cal L}_{so}$ that are orthogonal to two directions $y,z$ (that cannot coincide with $(1,\ldots, 1)$), i.e., the directions in a subspace of co-dimension $2$, which gives
$D_R = \Theta(2^{(k-2)n})$. Similarly, we obtain $D_L =  \Theta(2^{(k-1)n}) \cdot \Theta(2^{(k-2)n})$.

Observe that $\# H  = \#L \cdot D_L = \#R \cdot D_R =  \Theta(2^{kn}) \cdot \Theta(2^{(k-1)n}) \cdot \Theta(2^{(k-2)n})$.

Now we can compute the eigenvalues. 
We denote
$
M =\left(
\begin{array}{cc}
0 & A\\
A^\top &0
\end{array}
\right)
$
the adjacency matrix of the graph. 
We estimate the eigenvalues of $A\cdot A^\top $, which is the matrix of paths of length $2$ in the graph,  starting and finishing in $L$. 
Starting at some $x\in L$,  we can go to some $(y,z)\in R$ and then either come back to the same $x$, or to end up in a different $x'\in L$. 
For a fixed $x$, the number of paths 
\[
x \to (y,z) \to x
\]
is equal to $D_L$ (any $(y,z) $ matching $x$  serves as the middle point of the path). 
For $x\not= x'$,  the number of paths
\[
x \to (y,z)  \to x'
\]
is equal to the number of mutually orthogonal  pairs  $(y,z) $ that are both orthogonal to $x$ and $x'$, and all four directions $x, x', y, z$ are distinct. 
This is similar to the computation of $D_L$ but the co-dimensions are incremented: we have  $ \Theta(2^{(k-2)n}) \cdot \Theta(2^{(k-3)n})$ such pairs.

We denote by $I$ the identity matrix and  by $J$ the matrix of all ones. 
Those matrices are both symmetric.  We get
\[
\begin{array}{rcl}
A\cdot A^\top &=& D_L \ \cdot \ I  + \Theta(2^{(k-2)n} \cdot 2^{(k-3)n})\  \cdot\ (J-I)  \\
&=&   \Theta(2^{(k-1)n}  \cdot 2^{(k-2)n} ) \cdot I  + \Theta(2^{(k-2)n} \cdot 2^{(k-3)n}) \ \cdot \   J.
\end{array}
\]
Observe that $I$ and $J$ have a common basis of eigenvectors.
Indeed, for the matrix $I$ all vectors in the space are eigenvectors (and the only eigenvalue is $1$ with multiplicity $\# L$). 
The eigenvalues of $J$ are the number $\#L$ (of multiplicity $1$) and  $0$  (of multiplicity $\#L-1$).
Therefore, the eigenvectors of $A\cdot A^\top$ are
\[
\begin{array}{l}
\lambda_1 = \Theta(2^{(k-2)n} \cdot 2^{(k-3)n})  \cdot \#L = \Theta(2^{ k n} \cdot 2^{(k-2)n} \cdot 2^{(k-3)n}) = \Theta(2^{(3k-5)n})\\
\lambda_2 =\ldots = \lambda_{\# L} =  \Theta(2^{(k-1)n}  \cdot 2^{(k-2)n} )  = \Theta(2^{(2k-3)n}) = \Theta(D_L).
\end{array}
\]
(Observe that $\lambda_1$ can be found directly as $D_L\cdot D_R$.)
The eigenvalues of $M$ are the square roots of those of $A \cdot A^\top$.

If $k$ is even, the computation is similar, but on each step we should subtract from the set of self-orthogonal directions the vectors collinear with $(1,\ldots,1)$.
\end{proof}

\section{Useful information inequalities}\label{sec:inequalities}

\begin{proof}[Proof of Lemma~\ref{lemma:ineq2}]
We need to prove that 
\begin{equation}\label{eq:4}
\Inf(x:y\mid s)  \lep \Inf(x:y) + \Inf(s:xy).
\end{equation}
Observe that 
\[
\Inf(x:y)  \eqp \Inf(x:y\mid s) + \Inf(x:y:s)  
\]
and 
\[
\Inf(s:xy) \eqp \Inf(s:x) + \Inf(s:y) - \Inf(x:y:s).
\]
Therefore, \eqref{eq:4} rewrites to 
\[
I (x:y\mid s) \lep   \Inf(x:y\mid s) + \Inf(x:y:s)   +  \Inf(s:x) + \Inf(s:y) - \Inf(x:y:s).
\]
This inequality is always true since the terms $\Inf(s:x)$ and $\Inf(s:y)$ are non-negative.
\end{proof}

\begin{lemma}\label{lemma:appendix1}
 For all $x,y,x',y',r$

\noindent
(i) $
 \Inf(x':y \mid z) \lep  \Inf(x : y \mid z)  + \C(x' \mid x).
$

\noindent
(ii) $
 \Inf(x':y' \mid z) \lep  \Inf(x : y \mid z)  + \C(x' \mid x) + \C(y'\mid y).
$

\noindent
(iii) $
 \Inf(x':y' \mid z,r) \lep  \Inf(x : y \mid z,r)  + \C(x' \mid x,r) + \C(y'\mid y,r).
$

\end{lemma}
\begin{proof}
To prove (i), we observe that by the chain rule for the mutual information 
\[
\Inf(x,x':y \mid z) \eqp \Inf(x:y\mid z) +   \Inf(x':y\mid x,z) \eqp  \Inf(x':y\mid z) +   \Inf(x:y\mid x',z) .
\]
Therefore,
\[
\begin{array}{rcl}
\Inf(x':y \mid z) &\eqp & \Inf(x:y \mid z) + \Inf(x':y \mid x,z) - \Inf(x:y|x',z)\\
 &\lep & \Inf(x:y \mid z) + \Inf(x':y \mid x,z) \\
&\lep &  \Inf(x:y \mid z)  + \C(x'\mid x,z) \\
&\lep &  \Inf(x:y \mid z)  + \C(x'\mid x) .
\end{array}
\]
To prove (ii) we apply the same argument twice (at first we replace $x'$ by $x$ and then replace $y'$ by $y$),
\[
 \Inf(x':y' \mid z) \lep  \Inf(x : y' \mid z)  + \C(x' \mid x)  \lep  \Inf(x : y \mid z)  + \C(x' \mid x)   + \C(y' \mid y).  
\]
The proof of (iii) is a ``relativized'' version of the proof of (ii); we only need to add $r$ to the condition in all term of Kolmogorov complexity in the argument.
\end{proof}

\section{Lemma on relativization}\label{sec:relativize}

Kolmogorov complexity can be relativized: for any oracle $\cal O$, we may define Kolmogorov complexity $\C^{\cal O}(x)$, $\C^{\cal O}(x\mid y)$ in terms of a universal decompressor that can access $\cal O$. 
In case when $\cal O$ is represented by a finite string $s$, the relativization has a simple meaning: $\C^{\cal O}(x) = \C(x\mid s)+O(1)$ and $\C^{\cal O}(x\mid y) = \C(x\mid y, s)$.
Having fixed an oracle, we can pose the problem of secret key agreement in terms of the relativized Kolmogorov complexity.

The lower bounds on communication complexity of secret key agreement that we have discussed followed a pretty constructive scheme: 
we defined a set of input data sets $H$ such that 
\begin{itemize}
\item most triples $(x,y,z)\in H$ have complexity profile close to the required parameters, and
\item we show that any secret key agreement that succeeds on most $(x,y,z)\in H$ must have large communication complexity.  
\end{itemize}
This type of argument easily relativizes. Indeed, we can show that an oracle contains negligible information on  most $(x,y,z)\in H$; 
therefore, the relativization does not change significantly  the complexity profile for most triples in $H$.
It follows that a successful secret key agreement scheme for triples from $H$ in the sense of the relativized Kolmogorov complexity can be used as  a secret key agreement scheme  in the sense of the standard (non-relativized) Kolmogorov complexity.
Thus, if we had a lower bound on the communication complexity of a successful protocol in the non-relativized setting, substantially the same bound applies to the relativized version of the problem.
In other words, \emph{relativization cannot make the problem of secret key agreement easier} (cannot reduce communication complexity).

However, this argument does not imply that  relativization cannot make the communication complexity of the problem \emph{harder}.
Indeed,  the relativization conditional on an oracle $\cal O$ can provide us with completely  new tuples $(x,y,z)$ with the required complexity profile. 
These new input data sets  may have unusual combinatorial properties, and \emph{a priori} they might require a longer communication to agree on a secret key.
In what follows we show that this is not the case: for a fixed complexity profile, relativization cannot increase communication complexity of secret key agreement.

\bigskip

In this section we use the technique of ``clones'' that was developed in \cite{hammer,romashchenko-pairs,muchnik-romashchenko}. 
We need to recall several definitions.
For a tuple $(x_1,\ldots, x_n)$, its \emph{complexity profile} is the vector of $2^n-1$ values of Kolmogorov complexities $\C(x_{i_1},\ldots, x_{i_k})$ for all $1\le i_1<\ldots i_k \le n$;
the \emph{extended complexity profile} includes (besides the same  $2^n-1$ values of unconditional complexities) the vector of conditional complexities $\C(x_{i_1},\ldots, x_{i_k} \mid x_{j_1},\ldots, x_{j_s} )$ 
for all disjoint sets of indices $1\le i_1<\ldots i_k \le n$ and $1\le j_1<\ldots j_s \le n$. 

For a tuple $(x_1,\ldots, x_n)$, the set of its \emph{clones} denoted $\text{Clone}(x_1,\ldots,x_n)$ is defined as the set of all $(x_1',\ldots, x_n')$ such that the extended complexity profile of $(x_1',\ldots, x_n')$ is component-wise not greater than the extended complexity profile of $(x_1,\ldots, x_n)$. It is known (see \cite{romashchenko-pairs,muchnik-romashchenko}) that
\begin{itemize}
\item \ \textbf{substantiality:} $\log \left(\# \text{Clone}(x_1,\ldots,x_n)\right) \eqp \C(x_1,\ldots,x_n)$, and 
\item \ \textbf{uniformity:} if the set of all indices is split into two parts, 
\[
\{1,\ldots,n\} = \{i_1,\ldots,i_k\} \sqcup \{j_1,\ldots, j_s\},
\]
then for every set of strings $x'_{j_1},\ldots,x'_{j_s}$, there are at most $2^{\C(x_{i_1},\ldots, x_{i_k} \mid x_{j_1},\ldots, x_{j_s} ) +1}$ 
tuples $x'_{i_1},\ldots,x'_{i_k}$ such that the $n$-tuple combined of $x'_{j_1},\ldots,x'_{j_s}$ and $x'_{i_1},\ldots,x'_{i_k}$ belongs to $\text{Clone}(x_1,\ldots,x_n)$.
%Moreover, for most $(x'_{j_1},\ldots,x'_{j_s})$ (in the projection of the set $\text{Clone}(x_1,\ldots,x_n)$ onto the coordinates $j_1,\ldots, j_s$) this inequality is tight up to a factor of $2^{O(\log n)}$.
\end{itemize}
Similarly, we can define the extended complexity profile and the set of clones using Kolmogorov complexity relativized conditional on an oracle. 
For the relativized clones we have the same properties of  substantiality and uniformity.

\begin{lemma}\label{l:protocol-relativization}
Let $\pi$ be a communication protocol of secret key agreement for three participants (with inputs of length $n$), with public randomness ($m=O(n)$ public random bits).
Assume that there exists an oracle $\cal O$ and a complexity profile $\bar p \in \mathbb{N}^7$ such that for some
triple of inputs $(x,y,z)$ with
\[
(\C^{\cal O}(x), \C^{\cal O}(y), \C^{\cal O}(z), \C^{\cal O}(x,y), \C^{\cal O}(x,z), \C^{\cal O}(y,z), \C^{\cal O}(x,y,z)) = \bar p
\]
the protocol $\pi$ fails  (with a probability at least $1/2$) to obtain a common secret key. 
The secrecy failure means that given oracle $\cal O$,  Kolmogorov complexity of the key conditional on the transcript $\text{\rm transcript}_\pi(x,y,z,r)$  and the string $r$ sampled by the public source of random bits is below some threshold $\ell$,
\[
\C(\text{\rm produced key} \mid r, \text{\rm transcript}_\pi(x,y,z,r)) < \ell.
\]
Then there exists a triple of inputs $(x',y',z')$ whose non-relativized complexity profile 
\[
\C(x'), \C(y'), \C(z'), \C(x',y'), \C(x',z'), \C(y',z'), \C(x',y',z')
\]
 is component-wise $O(\log n)$-close to $\bar p$, and the protocol  $\pi$ fails  (also with a probability at least $1/2$) to obtain a common secret key 
 (here the secrecy is understood without relativization, with a threshold $\ell' = \ell+O(\log n)$).
 \end{lemma}
\begin{proof} 
We begin the proof with some notation.
Let $(x,y,z)$ be a triple of inputs from the statement of the theorem. 
We denote by $\bar q = \bar q(x,y,z)$ the extended complexity profile of this triple. 
For each triple of inputs $(x,y,z)$ and for each $m$-bit string $r$ we denote by 
\[
\text{transcript}_\pi(x,y,z,r)
\] 
the transcript of the protocol applied to these inputs with public random bits $r$ and by 
\[
\text{result}_\pi(x,y,z,r)
\]
the key produced by the parties (if the parties fail to  agree on a common key, we let  $w$ be the empty word).

If $\pi$ fails  (in the sense of the relativized Kolmogorov complexity) on some triple of inputs $(x,y,z) \in \{0,1\}^n\times \{0,1\}^n \times \{0,1\}^n$, it means that for a half of all bit strings $r\in \{0,1\}^m$, 
when the protocol is applied to these inputs,
\begin{itemize}
\item[(i)] either Alice, Bob, and Charlie fail to  agree on a common  key, 
\item[(ii)] or they do agree on one and the same key but this key is not secret, i.e., its relativized complexity is below the threshold,  
$
\C^{\cal O}(\text{result}_\pi(x,y,z,r) \mid \text{transcript}_\pi(x,y,z,r), r) <\ell.
$
\end{itemize}

We denote by $F \subset  \{0,1\}^n\times \{0,1\}^n \times \{0,1\}^n $ the set of all triples of inputs $(x,y,z)$ on which $\pi$ fails, and 
$
F':= F\cap {\rm Clone}(\bar q),
$
i.e., the set of all failure inputs that have extended complexity profile component-wise below~$\bar q$. 

We say that a tuple of strings
\begin{equation}
\label{eq:certificate}
(x,y,z,r,\text{transcript}_\pi(x,y,z,r),\text{result}_\pi(x,y,z,r)) 
\end{equation}
is a \emph{positive certificate}  for $(x,y,z)$ if all three  parties agree on the same key and
\begin{equation}
%\label{eq:secrecy}
\nonumber
\C^{\cal O}(\text{result}_\pi(x,y,z,r) \mid r, \text{transcript}_\pi(x,y,z,r)) \ge \ell.
\end{equation}
Otherwise, this tuple is called a \emph{negative certificate} for $(x,y,z)$. 
By definition, $\pi$ fails on $(x,y,z)$ if  for a half of all bit strings $r\in \{0,1\}^m$, the tuple \eqref{eq:certificate} is a negative certificate.

We say that $w$ is \emph{compatible} with $(r,t)$, if there exists at least one  negative certificate $(x,y,z,r,t,w)$. 
Observe that for each $r$ and $t$ there are less than $2^{\ell}$ strings $w$ such that 
\[
\C^{\cal O}(w \mid r, t) < \ell.
\]
Therefore, for each $(r,t)$ there are less than $2^\ell$ strings $w$ compatible with $(r,t)$.

We know that $F'$ is not empty (by the condition of the lemma, at least one triple $(x,y,z)$ causes a failure of $\pi$). Observe that 
\begin{itemize}
\item $F'$ inherits the property of  \emph{uniformity} from $\text{Clone}(x,y,z)$,
\item we can enumerate the set $F'$ given access to $\cal O$,
\item since $(x,y,z)$ can be specified by its ordinal number in this enumeration,
\[
\C^{\cal O}(x,y,z) \lep \log \#F' 
\]
or, equivalently,
$
\# F' \ge 2^{\C^{\cal O}(x,y,z) - O(\log n)}.
$
\end{itemize}
%For most triples in $F'$ their extended complexity profile is not only component-wise smaller than $q(x,y,z)$ but $O(\log n) -close to $q(x,y,z)$.
Now define a combinatorial structure that looks ``similar'' to the triple of sets
\begin{equation}
\label{eq:structure}
(F', \text{positive certificates for }F', \text{negative certificates for }F'),
\end{equation}
but the oracle $\cal O$ is not involved in the definition.
Having fixed an arbitrary set  of tuples $G\subset  \{0,1\}^n\times \{0,1\}^n \times \{0,1\}^n$ and the protocol $\pi$, we can compute the set $S(G)$ of all tuples \eqref{eq:certificate} for every $(x,y,z)\in G$. 
We consider possible splits of $S$ into two classes,
$
S = S_{+} \sqcup S_{-}
$
and say that $(r,t)$ is $G$-compatible with $w$, if there is at least one  tuple $(x,y,z,r,t,w)\in S_-$ (for some $(x,y,z)\in G$).
Such a split is \emph{valid} if 
\begin{itemize}
\item[(a)] for each $(x,y,z) \in G$ and for a half of $r\in \{0,1\}^m$, the tuple 
\[
(x,y,z,r,\text{transcript}_\pi(x,y,z,r),\text{result}_\pi(x,y,z,r))
\] 
belongs to $S_-$, and
\item[(b)] for each $(r,t)$ there are less than $2^\ell$ strings $w$ that is $G$-compatible with $(r,t)$.
\end{itemize}

We say that  a valid split  $(G, S_+, S_-)$ is \emph{similar} to the original triple \eqref{eq:structure} if
\begin{itemize}
\item \textbf{a variant of substantiality:} $\log \# G \gep \C^{\cal O}(x,y,z)$,
\item \textbf{a variant of uniformity:}
the cardinalities of sections and projections of $G$ are not greater than the exponent of the corresponding complexity term in  $\bar q$
(similarly to the property of uniformity of $F'\subset \text{Clone}^{\cal O}(x,y,z)$ as formulated above).
\end{itemize}
Obviously, triples of sets satisfying the definition of  similarity exist, e.g.,  the original set $F'$ with its positive and negative certificates is similar to itself.
Let $(G^0, S^0_+, S^0_-)$ be the very first (e.g., in the lexicographical order) triple of sets respecting the presented requirements.  

Although we cannot enumerate the elements of $F'$ without access to the oracle $\cal O$, the sets $(G^0, S^0_+, S^0_-)$ can be found algorithmicaly given only the protocol $\pi$ and the components of the extended vector $\bar q$.
Observe that for most triples $(x',y',z')\in G_0$, the extended complexity profile (non-relativized one) is $O(\log n)$-close to $\bar q$. 
Condition~(b) above implies that for each $w$ that is $G^0$-compatible with $(r,t)$ we have $\C(w\mid r,t)\lep \ell$ (once again, this Kolmogorov complexity term is not relativized).

Thus, most $(x',y',z')\in G_0$ have the required (non-relativized) complexity profile, and  by the construction of triples similar to \eqref{eq:structure}, our protocol $\pi$ fails on these triples  with a probability $>1/2$. 
This observation concludes the proof.
\end{proof}
\begin{remark}
If a communication protocol leaks minor information on the key to the adversary,
\[
\C(\text{key} \mid \text{transcript},\text{public randomness})  = |\text{key}| - \delta,
\]
 we can improve the secrecy by taking a random hash $\text{key}' = \text{hash} (\text{key},\text{public random bits})$ such that 
the new $\text{key}'$ is $\delta+O(1)$ shorter than original $\text{key}$ but
\[
\C(\text{key}' \mid \text{transcript},\text{public randomness})  = |\text{key}'| - O(1),
\]
see \cite{jacm2019}. Thus, if  we can agree on a mildly secure key of an asymptotically optimal size, 
we can also agree on a strongly secure key of  approximately the same size.
\end{remark}

\begin{remark}
Lemma~\ref{l:protocol-relativization} can be understood in the counter-positive way: 
if Alice, Bob, and Charlie can efficiently agree on a secret key for triples of inputs $(x,y,z)$ with some specific complexity profile (in the sense of the standard non-relativized Kolmogorov complexity),
they can do the same in the sense of Kolmogorov complexity relativized conditional on oracle $\cal O$.
In other words, relativization changes the set of inputs with a specified complexity profile but it does not make the problem of secret key agreement more difficult.
\end{remark}

\section{Proof of Lemma~\ref{l:reduction}}\label{sec:reduction}
Let us recall Muchnik's theorem on conditional descriptions and its version proven  by Bauwens and Zimand.
\begin{theorem}\label{th:muchnik}
(a) \cite{muchnik-coding}
For every string $a$ and for all strings $b_1,\ldots,b_\ell$ and for every number $m\le \C(a)$ 
there exists a ``digital fingerprint'' of $a$ of length $m$ that looks maximally random conditional on each $b_j$. 
Technically, this means that for some $\tilde a$ we have
\[
\C(\tilde a) \eqp m,\ \C(\tilde a \mid a) \eqp 0, \ \text{and} \ \C( \tilde a\mid b_j) = \min\{ \C( a\mid b_j), m\} \ \text{for}\ j=1,\ldots,\ell. 
\]
(b) \cite{bauwens-coding,zimand-coding} (see also the \emph{Single Source Compression Theorem} in \cite{jacm2019}) Moreover, such a ``fingerprint''  can be constructed pretty explicitly: given the length of string $a$, the numbers $\C(a\mid b_j)$, and $m$, one can construct an algorithm $\text{\rm Code}$ such that the conditions from (a) are valid for the vast majority of strings $\tilde a = \text{\rm Code}(a,r)$, where the probability is taken over the choice of a string $r$ of length $O(\log (|a|+|b_1|+\ldots+|b_\ell |))$.
\end{theorem}
\begin{proof}[Proof of Lemma~\ref{l:reduction}]

\noindent
\textbf{Proof of~(a).} We apply Theorem~\ref{th:muchnik} for $\ell=1$, with $a=x$  and $b_1=\langle y,z\rangle$,  and $m = \delta$ and obtain a string $\tilde x$ such that 
\[
\C(\tilde x) \eqp \delta,\ \C(\tilde x \mid x) \eqp 0, \  \C(\tilde x \mid y,z) \eqp \delta.
\]
In a similar way (applying again Theorem~\ref{th:muchnik}) we obtain $\tilde y$ and $\tilde z$ such that 
\[
\C(\tilde y) \eqp \delta,\ \C(\tilde y \mid y) \eqp 0, \  \C(\tilde y \mid x,z) \eqp \delta
\text{ and }
\C(\tilde z) \eqp \delta,\ \C(\tilde z \mid z) \eqp 0, \  \C(\tilde z \mid x,y) \eqp \delta.
\]
A routine check shows that the triple $(x,y,z)$ \emph{conditional on}  $\langle \tilde x, \tilde y, \tilde z\rangle $ has complexity profile \eqref{eq:reduction-a}. % and $\Inf(x'y'z' : \tilde x \tilde y \tilde z) \eqp0$.

Let $\pi'$ be a communication protocol for the profile \eqref{eq:reduction-a}.
Given $(x,y,z)$, Alice, Bob, and Charlie can do as follows: each of them computes a fingerprint $\tilde x, \tilde y, \tilde z$ for $x$, $y$, and $z$ respectively, and  broadcast them. 
Then they proceed with a protocol $\pi'$ applied to $(x,y,z)$ and the complexity profile \eqref{eq:reduction-a} (for Kolmogorov complexity relativized conditional on $\tilde x, \tilde y, \tilde z$). 
If the protocol succeeds, they obtain a secret key $w$ that is incompressible given the public random bits and the full transcript of the combined protocol, which consists of the transcript of $\pi'$, the public random bits, \emph{and the broadcasted strings $\tilde x, \tilde y, \tilde z$}.
Thus, we obtain a communication protocol for the original $(x,y,z)$, whose communication complexity is equal to the communication complexity of $\pi'$ increased by $3\delta$ bits  (the total length of $\tilde x, \tilde y, \tilde z$ broadcasted at the first stage).

We know the optimal communication complexity of secret key agreement for the original $(x,y,z)$  due to Theorem~\ref{th:main-special-case}. 
Therefore, communication complexity of $\pi'$ (for the relativized complexity profile) cannot be better than \eqref{eq-cc-3}, 
which is in our case $3\alpha +\frac32\beta-3\delta$ (the optimal communication complexity for $(x,y,z)$ decreased by $3\delta$).
It remains to use Lemma~\ref{l:protocol-relativization} and conclude that the communication complexity of secret key agreement for  $(x',y',z')$ having the non-relativized complexity profile  \eqref{eq:reduction-a} cannot be better than  $3\alpha +\frac32\beta - 3\delta$.

\medskip

\noindent
\textbf{Proof of~(b).} Let $v$  be a string of $\delta$ bits  such that $I(x,y,z:v)\eqp 0$, and 
\begin{equation} \label{eq:common-block}
x' = \langle x,v\rangle,\ y' = \langle y,v\rangle, \ z' = \langle z,v\rangle.
\end{equation}
It is easy to verify that complexity profile of $(x',y',z')$ matches \eqref{eq:reduction-b}.
From Theorem~\ref{th:algorithmic-3} it follows that the optimal size of a secret key that three parties can agree on when given $x',y',z'$ as inputs is
\begin{equation} \label{eq:longer-key}
\begin{array}{l}
\frac12\left(I(x':y'\mid z') +  I(x':z'\mid y') + I(y':z'\mid z')  \right) + I(x':y':z')   \\
   \eqp \frac32\beta+\gamma+\delta  
  \eqp \frac12\left(I(x:y\mid z) +  I(x:z\mid y) + I(y:z\mid z)  \right) + I(x:y:z) + \delta.
\end{array}
\end{equation}
(i.e., the  size of the key for the triple of inputs $(x,y,z)$ plus $\delta$).
Our aim is to show that such a protocol requires communication complexity at least 
 \[
\C(x',y',z') - \frac12\big(I(x':y'\mid z') +  I(x':z'\mid y') + I(y':z'\mid z')  \big)  - I(x':y':z') = 3 \alpha + \frac32 \beta.
 \]
 Assume for the sake of contradiction that there is protocol $\pi$ that achieves the goal with communication complexity $3 \alpha + \frac32 \beta - \epsilon$.
 In what follows we construct a protocol $\pi'$ that allows to construct an optimal size secret key with the same communication complexity for the original inputs $(x,y,z)$.

In the new protocol Alice, Bob, and Charlie take from the (common) public source of random bits a string of $\delta$ bits $v$  and define $x',y',z'$ as in \eqref{eq:common-block}.
With an overwhelming probability we have $I(x,y,z:v)\eqp 0$, so we have \eqref{eq:reduction-b}. Then Alice, Bob, and Charlie proceed as in protocol $\pi$ and find a common key $w$ of size \eqref{eq:longer-key}.
As $\pi$ is a valid protocol of secret key agreement, the key $w$ has zero mutual information with the transcript $t$ of the protocol. 
However, we loose the conditional  of secrecy when $v$ is public (which is in our case a part of the public source of random bits accessible to the attacker).
However, we may restore the secrecy by reducing the size of the key.
We apply Theorem~\ref{th:muchnik} and construct $w' = \text{Code}(w,r')$ (where $r'$ is a string of public random bits of logarithmic size) such that
\[
\C(w') = \frac32\beta+\gamma, \ \Inf(w' : v, t) \eqp0
\]
with a high probability (over the choice of $r'$).
The produced $w'$ can be taken as a secret key. The new protocol has communication complexity $3 \alpha + \frac32 \beta - \epsilon$ (the same as $\pi$),
and we get a contradiction with Theorem~\ref{th:main-special-case} unless $\epsilon\eqp0$.

\medskip

\noindent
\textbf{Proof of~(c).}
Let $(x,y,z)$ be a hyperedge of tri-expander, as in the proof of Theorem~\ref{th:main-special-case}.
We will transform this triple in a different triple of inputs $(x',y',z')$ using the following trick suggested by Alexander Shen,~\cite{personal-communication}.
We apply Theorem~\ref{th:muchnik} with $\ell = 3$ for $a=x$, $b_1 = y$, $b_2 = z$, $b_3 = \langle y,z\rangle$ and $m = \C(x) - \delta$, and obtain an $x'$ such that
\[
\C(x') = \C(x) - \delta, \ \C(x' \mid x) \eqp 0,\ \C(x'\mid y) \eqp \C(x\mid y),\ \C(x'\mid z) \eqp \C(x\mid z),\ C(x'\mid y,z) \eqp \C(x\mid y,z).
\]
In a similar way, we obtain $y'$ and $z'$ such that 
\[
\C(y') = \C(y) - \delta, \ \C(y' \mid y) \eqp 0,\ \C(y'\mid x) \eqp \C(y\mid x),\ \C(y'\mid z) \eqp \C(y\mid z),\ C(y'\mid x,z) \eqp \C(y\mid x,z)
\]
and
\[
\C(z') = \C(z) - \delta, \ \C(z' \mid z) \eqp 0,\ \C(z'\mid x) \eqp \C(z\mid x),\ \C(z'\mid y) \eqp \C(z\mid y),\ C(z'\mid x,y) \eqp \C(z\mid x,y).
\]
It is not hard to verify that the triple $(x',y',z')$ has complexity  profile \eqref{eq:reduction-c}, see Fig.~\ref{fig:reduced-profiles}.
			\begin{figure}[h]

    			 \begin{subfigure}[b]{0.49\textwidth}
			 	\begin{flushleft}
				\begin{tikzpicture}[scale=0.6]
				  \draw \firstcircle node[above left] {\small $\alpha$};
				  \draw \secondcircle node [above right] {\small $\alpha$};
				  \draw \thirdcircle node [below] {\small $\alpha$};
				  \node at (95:0.05)   {\small $0$};
				  \node at (90:1.55) {\small $ \beta$};
				  \node at (210:1.75) {\small $\beta$};
				  \node at (330:1.75) {\small $\beta$};
				  \node at (150:4.75) {\Large $x$};
				  \node at (30:4.75) {\Large $y$};
				  \node at (270:4.75) {\Large $z$};
				\end{tikzpicture}
				\caption{$\C(x\mid y,z)\eqp  \C(y\mid x,z) \eqp \C(z\mid x,y) \eqp \alpha$, \newline 
				$\Inf(x:y\mid z) \eqp \Inf(x:z\mid y) \eqp \Inf(y:z\mid x) \eqp \beta$, \newline  $\Inf(x:y:z) \eqp0$.}
				\end{flushleft}
			\end{subfigure}
			\hspace{2em}	
			 \begin{subfigure}[b]{0.49\textwidth}
			 	\begin{flushleft}
				\begin{tikzpicture}[scale=0.6]
				  \draw \firstcircle node[above left] {\small $\alpha$};
				  \draw \secondcircle node [above right] {\small $\alpha$};
				  \draw \thirdcircle node [below] {\small $\alpha$};
				  \node at (95:0.05)   {\small $-\delta$};
				  \node at (90:1.55) {\small $ \beta$};
				  \node at (210:1.75) {\small $\beta$};
				  \node at (330:1.75) {\small $\beta+\delta$};
				  \node at (150:4.75) {\Large $x'$};
				  \node at (30:4.75) {\Large $y$};
				  \node at (270:4.75) {\Large $z$};
				\end{tikzpicture}
				\caption{$\C(x' | y,z)\eqp  \C(y | x',z) \eqp \C(z | x',y) \eqp \alpha$, \newline 
				$\Inf(x':y | z) \eqp \Inf(x':z | y) \eqp \beta$, $ \Inf(y:z | x') \eqp \beta+\delta $, \newline  $\Inf(x':y:z) \eqp -\delta$.}
				\end{flushleft}
			\end{subfigure}

	\vspace{2em}
	
	    			 \begin{subfigure}[b]{0.49\textwidth}
			 	\begin{flushleft}
				\begin{tikzpicture}[scale=0.6]
				  \draw \firstcircle node[above left] {\small $\alpha$};
				  \draw \secondcircle node [above right] {\small $\alpha$};
				  \draw \thirdcircle node [below] {\small $\alpha$};
				  \node at (95:0.05)   {\small $-2\delta$};
				  \node at (90:1.55) {\small $ \beta$};
				  \node at (210:1.75) {\small $\beta +\delta $};
				  \node at (330:1.75) {\small $\beta+\delta $};
				  \node at (150:4.75) {\Large $x'$};
				  \node at (30:4.75) {\Large $y'$};
				  \node at (270:4.75) {\Large $z$};
				\end{tikzpicture}
				\caption{$\C(x'| y',z)\eqp  \C(y'| x',z) \eqp \C(z | x',y') \eqp \alpha$, \newline 
				$\Inf(x':y' | z) \eqp \beta $, $  \Inf(x':z| y') \eqp \Inf(y':z| x') \eqp \beta+\delta$, \newline  $\Inf(x':y':z) \eqp -2 \delta$.}
				\end{flushleft}
			\end{subfigure}
			\hspace{2em}	
			 \begin{subfigure}[b]{0.49\textwidth}
			 	\begin{flushleft}
				\begin{tikzpicture}[scale=0.6]
				  \draw \firstcircle node[above left] {\small $\alpha$};
				  \draw \secondcircle node [above right] {\small $\alpha$};
				  \draw \thirdcircle node [below] {\small $\alpha$};
				  \node at (95:0.05)   {\small $-3\delta$};
				  \node at (90:1.55) {\small $ \beta+\delta$};
				  \node at (210:1.75) {\small $\beta+\delta$};
				  \node at (330:1.75) {\small $\beta+\delta$};
				  \node at (150:4.75) {\Large $x'$};
				  \node at (30:4.75) {\Large $y'$};
				  \node at (270:4.75) {\Large $z'$};
				\end{tikzpicture}
				\caption{$\C(x' | y',z')\eqp  \C(y' | x',z') \eqp \C(z'| x',y') \eqp \alpha$, \newline 
				$\Inf(x':y' | z') \eqp \Inf(x':z'| y') \eqp \Inf(y':z'| x') \eqp \beta+\delta $, \newline  $\Inf(x':y':z') \eqp -3\delta$.}
				\end{flushleft}
			\end{subfigure}
			\caption{Complexity profile for Muchnik's fingerprints of $x,y,z$.}
				\label{fig:reduced-profiles}
			\end{figure}

From Theorem~\ref{th:algorithmic-3} it follows that the optimal size of a secret key that three parties can agree on when given $x',y',z'$ as inputs is
\[
\frac12\left(I(x':y'\mid z') +  I(x':z'\mid y') + I(y':z'\mid z')  \right) + I(x':y':z') \eqp \frac{3 \beta}2 - \frac{3\delta}2 .
\]
This means  (see Corollary~\ref{col:1}) that Charlie must send a message $m_C$ such that  
\[
\Inf(x':y' \mid m_C) \gep \frac{3 \beta}2 - \frac{3\delta}2 .
\]
Observe that $\Inf(x':y') \eqp \beta - 2\delta$.  
Therefore, the mutual information between $m_C$ and $(x',y')$ must be greater than the difference between $\Inf(x':y' \mid m_C)$ and $\Inf(x':y')$,
\[
 \Inf(m_C : x',y') \gep \left( \frac{3 \beta}2 - \frac{3\delta}2 \right)- \left(\beta - 2\delta\right) \eqp \frac\beta2 + \frac\delta2
\]
(as in the proof of Theorem~\ref{th:main-special-case}). If $(x',y',z')$ are obtained from a hyperedge of a tri-expander, then we can apply Corollary~\ref{thm:1-corollary} and conclude that 
\[
\C(m_C) \gep \C(z' \mid x'y') +   \frac\beta2 + \frac\delta2 = \alpha +  \frac\beta2 + \frac\delta2 .
\]
A similar argument gives
\[
\C(m_A) \ge  \alpha +  \frac\beta2 + \frac\delta2  \text{ and }
\C(m_B) \ge \alpha +  \frac\beta2 + \frac\delta2 
\]
for the messages sent by Alice and Bob respectively. By summing up these bounds we conclude that the total length of all messages must be at least
\[
 3\alpha +  \frac{3 \beta}2 + \frac{3\delta}2, 
\]
which is exactly the communication complexity of the omniscience protocol.
\end{proof}

\end{document}